\documentclass[11pt]{article} 
\usepackage{rpmacros}
\usepackage[dvipsnames,svgnames,x11names,hyperref]{xcolor}
\usepackage{amsthm}
\usepackage{xfrac}
\usepackage[T1]{fontenc}

\usepackage{gitinfo}
\usepackage{enumitem}
\usepackage{fullpage}
\usepackage{thmtools}
\usepackage{thm-restate}
\usepackage{nicefrac}
\usepackage[colorlinks,pagebackref]{hyperref}
\hypersetup{
  linkcolor=[rgb]{0,0,0.4},
  citecolor=[rgb]{0, 0.4, 0},
  urlcolor=[rgb]{0.6, 0, 0}
}
\usepackage{mathrsfs,bbm}
\usepackage{todonotes}
\usepackage{tikz}
\usetikzlibrary{decorations.pathreplacing,backgrounds}
\usepackage{multirow}
\usepackage{setspace}

\usepackage{amsthm}
\usepackage{thmtools,thm-restate}
\usepackage[nameinlink,capitalise]{cleveref}

\numberwithin{equation}{section}
\declaretheoremstyle[bodyfont=\it,qed=\qedsymbol]{noproofstyle}

\declaretheorem[name=Observation,numbered=no]{observation*}

\declaretheorem[numberlike=equation]{theorem}

\declaretheorem[name=Theorem,numbered=no]{theorem*}

\declaretheorem[numberlike=equation]{lemma}
\declaretheorem[name=Lemma,numbered=no]{lemma*}

\declaretheorem[name=Corollary,numbered=no]{corollary*}

\declaretheorem[numberlike=equation]{proposition}
\declaretheorem[name=Proposition,numbered=no]{proposition*}

\declaretheorem[numberlike=equation,name=Claim]{claim}
\declaretheorem[name=Claim,numbered=no]{claim*}

\declaretheorem[name=Conjecture,numbered=no]{conjecture*}

\declaretheorem[name=Question,numbered=no]{question*}

\declaretheoremstyle[bodyfont=\it,qed=$\lozenge$]{defstyle} 

\declaretheorem[numberlike=equation,style=defstyle]{definition}
\declaretheorem[unnumbered,name=Definition,style=defstyle]{definition*}

\declaretheorem[unnumbered,name=Example,style=defstyle]{example*}

\declaretheorem[unnumbered,name=Notation=defstyle]{notation*}

\declaretheorem[unnumbered,name=Construction,style=defstyle]{construction*}

\declaretheorem[numberlike=equation,style=defstyle]{remark}
\declaretheorem[unnumbered,name=Remark,style=defstyle]{remark*}

\newcommand{\LDT}{\operatorname{LDT}}
\newcommand{\Res}{\operatorname{Res}} 
\usepackage{nth}
\usepackage{intcalc}
\usepackage{etoolbox}
\usepackage{xstring}

\usepackage{ifpdf}
\ifpdf
\else
\usepackage[quadpoints=false]{hypdvips}
\fi

\newcommand{\ECCCurl}[2]{http://eccc.hpi-web.de/report/\ifnumcomp{#1}{>}{93}{19}{20}#1/#2/}

\newcommand{\shortECCC}[2]{\texttt{\href{http://eccc.hpi-web.de/report/\ifnumcomp{#1}{>}{93}{19}{20}#1/#2/}{eccc:TR#1-#2}}}

\newcommand{\parseECCC}[1]{
\StrSubstitute{#1}{TR}{}[\tmpstring]%
\IfSubStr{\tmpstring}{/}{ 
\StrBefore{\tmpstring}{/}[\ecccyear]%
\StrBehind{\tmpstring}{/}[\ecccreport]%
}{
\StrBefore{\tmpstring}{-}[\ecccyear]%
\StrBehind{\tmpstring}{-}[\ecccreport]%
}%
\shortECCC{\ecccyear}{\ecccreport}}

\newcommand*\samethanks[1][\value{footnote}]{\footnotemark[#1]}

\newif\ifdraft
\draftfalse
\ifdraft
\newcommand{\phnote}[1]{\todo[color=red!100!green!33,size=\footnotesize]{ph: #1}}
\newcommand{\sriprahladhuvacha}[1]{\todo[color=red!100!green!33,inline,size=\small]{ph: #1}}
\newcommand{\RPnote}[1]{\textcolor{BrickRed}{\guillemotleft RP: #1 \guillemotright}}
\newcommand{\MKnote}[1]{\textcolor{Orange}{\guillemotleft MK: #1 \guillemotright}}
\newcommand{\MSnote}[1]{\textcolor{Blue}{\guillemotleft MS: #1 \guillemotright}}
\newcommand{\gitinfonotecolour}{Gray}
\newcommand{\easteregg}{}
\else
\newcommand{\phnote}[1]{}
\newcommand{\sriprahladhuvacha}[1]{}
\newcommand{\RPnote}[1]{}
\newcommand{\MKnote}[1]{}
\newcommand{\MSnote}[1]{}
\newcommand{\gitinfonotecolour}{white}
\newcommand{\easteregg}{(What's the point of this line?)}
\fi

\newcommand{\vzero}{\mathbf{0}}
\newcommand{\ignore}[1]{}
\newcommand{\gitinfonote}{git info:~\gitAbbrevHash\;,\;(\gitAuthorIsoDate)\; \;\gitVtag}

\newcommand{\calL}{\mathscr{L}}
\newcommand{\calE}{\mathscr{E}}
\newcommand{\calP}{\mathscr{P}}
\newcommand{\calF}{\mathscr{F}}
\newcommand{\cF}{\calF}

\let\agree\agr
\DeclareMathOperator{\nonunique}{nonunique}

\renewcommand{\epsilon}{\varepsilon}
\newcommand{\eps}{\epsilon}
\newcommand{\corr}{\mathrm{corr}}
\newcommand{\BAD}{\text{BAD}}
\DeclareMathOperator{\Var}{Var}
\newcommand{\ehref}[1]{\href{mailto:#1}{#1}}

\newcommand{\hpartial}{{\mathchar'26\mkern-10mu \partial}}
\newcommand{\disc}{\mathsf{Disc}}

\crefname{claim}{Claim}{Claims}

\title{An Improved Line-Point Low-Degree Test}
\author{{Prahladh Harsha\thanks{Tata Institute of Fundamental Research, Mumbai, India. \ehref{prahladh@tifr.res.in, mrinal@tifr.res.in, ramprasad@tifr.res.in}.}}
\and
 {Mrinal Kumar\samethanks}
 \and
 {Ramprasad Saptharishi\samethanks}
 \and
 {Madhu Sudan\thanks{School of Engineering and Applied Sciences,
    Harvard University, Cambridge, MA, USA. Supported in part by a Simons Investigator Award and NSF Award CCF 2152413. \ehref{madhu@cs.harvard.edu}.}}}

\date{}

\begin{document}

\maketitle

\begin{abstract}
    We prove that the most natural low-degree test for polynomials over finite fields is ``robust'' in the high-error regime for linear-sized fields. Specifically we consider the ``local'' agreement of a function $f\colon \F_q^m \to \F_q$ from the space of degree-$d$ polynomials, i.e., the expected agreement of the function from univariate degree-$d$ polynomials over a randomly chosen line in $\F_q^m$, and prove that if this local agreement is $\epsilon \geq \Omega((\nicefrac{d}{q})^\tau))$ for some fixed $\tau > 0$, then there is a global degree-$d$ polynomial $Q\colon\F_q^m \to \F_q$ with agreement nearly $\epsilon$ with $f$. This settles a long-standing open question in the area of low-degree testing, yielding an $O(d)$-query robust test in the ``high-error'' regime (i.e., when $\epsilon < \nicefrac{1}{2}$). The previous results in this space either required $\epsilon > \nicefrac{1}{2}$ (Polishchuk \& Spielman, STOC 1994), or $q = \Omega(d^4)$ (Arora \& Sudan, Combinatorica 2003), or needed to measure local distance on $2$-dimensional ``planes'' rather than one-dimensional lines leading to $\Omega(d^2)$-query complexity (Raz \& Safra, STOC 1997). 

    Our analysis follows the spirit of most previous analyses in first analyzing the low-variable case ($m = O(1)$) and then ``bootstrapping'' to general multivariate settings. Our main technical novelty is a new analysis in the bivariate setting that exploits a  previously known connection between multivariate factorization and finding (or testing) low-degree polynomials, in a non ``black-box'' manner. This connection was used roughly in a black-box manner in the work of Arora \& Sudan --- and we show that opening up this black box and making some delicate choices in the analysis leads to our essentially optimal analysis. A second contribution is a bootstrapping analysis which manages to lift analyses for $m=2$ directly to analyses for general $m$, where previous works needed to work with $m = 3$ or $m = 4$ --- arguably this bootstrapping is significantly simpler than those in prior works. 
\end{abstract}

\newpage 

\tableofcontents

\newpage

\section{Introduction}

In this paper we consider the classical ``line-point'' ``low-degree test'' for multivariate polynomials over finite fields and give a ``near-optimal'' analysis in the ``high-error'' regime. We expand on these terms below.

The basic question in ``low-degree testing'' is to estimate the distance of a function $f\colon \F_q^m \to \F_q$ given as an oracle from the space $m$-variate polynomials over $\F_q$ of degree at most $d$. Broadly, for a family of function $\cF \subseteq \{f\colon D \to R\}$, a $t$-query $\alpha(\cdot)$-robust test for $\cF$ is a distribution supported on sets $S \subseteq D$ with $|S| \leq t$ such that for every function $f\colon D\to R$ we have $\delta(f,\cF) \leq \alpha(\E_S[\delta(f|_S,\cF_S)])$. (Here we use $\delta(f,g)$ to denote the normalized Hamming distance between functions, $\delta(f,\cF) := \min_{P \in \cF}\{\delta(f,P)\}$, $f|_S$ to denote the restriction of $f$ to the domain $S$ and $\cF|_S$ to denote $\{P|_S | P \in \cF\}$.) The quantity $\E_S[\delta(f|_S,\cF_S)]$ is thus a local measure of the distance of $f$ from $\cF$ and $\alpha(\cdot)$-robustness relates this local distance to the global distance of $f$ from $\cF$. Low-degree testing is the specialization of the study of robust testing in the case where $\cF = \F_q^d[X_1,\ldots,X_m]$, namely the space of $m$-variate polynomials of degree at most $d$ over $\F_q$, viewed as functions from $\F_q^m$ to $\F_q$. 
The {\em error (tolerance)} of a test is roughly this highest value of $\delta$ such that $\alpha(\delta)$ is (noticeably) bounded away from $1$. (Specifically in the case of low-degree testing we would want $\alpha(\delta) < 1 - \sqrt{\nicefrac{d}{q}}$.) For most natural tests, it is straightforward to actually show $\alpha(\delta) \approx \delta$ as long as it is noticeably bounded away from $1$. So the critical parameter desribing $\alpha$ is just the error-tolerance and this is what we will use to describe the history (and importance) of low-degree testing.

The low-degree testing problem is a classical problem with enormous impact in the early as well as state-of-the-art constructions of probabilistically checkable proofs (PCPs). It was introduced by Rubinfeld and Sudan~\cite{RubinfeldS1996} who, in our language, gave a $q$-query test that had an error-tolerance of $O(\nicefrac{1}{d})$ provided $q = \Omega(d^2)$. The test they introduced is now called the ``lines-point'' test and is given by the uniform distribution on lines in $\F_q^m$. Arora, Lund, Motwani, Sudan and Szegedy~\cite{AroraLMSS1998}, building on the work of Arora and Safra~\cite{AroraS1998}, improved the error-tolerance of the lines-point test to $\Omega(1)$ (provided $q = \Omega(d^3)$) and this was a crucial ingredient in getting PCPs with $O(1)$-query complexity. 
Subsequent improvements to PCP parameters were also closely related to improvements to (analyses of) low-degree tests. For instance the first constructions of nearly-linear sized PCPs by Polishchuk and Spielman~\cite{PolishchukS1994} rely on getting  error $= \Omega(1)$ for $q = O(d)$ and $m=O(1)$ in the lines-point test. This was extended to general $m$ by Friedl and Sudan~\cite{FriedlS1995} which till this work remains the optimal analysis for the lines-point test when $q = O(d)$. 

Higher error low-degree tests and analyses, with error tending to $1$, were obtained by Raz and Safra~\cite{RazS1997} and Arora and Sudan~\cite{AroraS2003}. The former introduced the ``planes-point test'', where the underlying distribution is uniform on ``planes'' (i.e., 2-dimensional affine subspaces of $\F_q^m$) and showed that it had error $1 - (\nicefrac{d}{q})^{\tau}$ for constant $\tau$. Thus, this obtains essentially optimal error, but at the cost of $q^2$-queries. The latter (\cite{AroraS2003}) improved the analysis of the line-point test, but only in the regime where $q$ was super-quartic in $d$ and showed roughly that it had error $1 - (\nicefrac{d^4}{q})^{\Omega(1)}$. So this reduced the query complexity to $O(q)$, but only when $q = \Omega(d^4)$. Thus, the three results above are essentially incomparable and represent the the three state-of-the-art low-degree tests today. (We remark there is also a vast body of related questions starting with the work of Alon, Kaufman, Krivelevich, Litsyn and Ron~\cite{AlonKKLR2005} that might be termed ``moderate degree testing'' where one considers the setting $q < d \ll m$. These results and their motivations are quite distinct from those in this work and we do not cover those results here. A direction of study more related to the setting of this paper (i.e., when $d \ll q$) involves derandomizing the low-degree test~\cite{BenSassonSVW2003,MoshkovitzR2008}. This direction turns out to be crucial in getting PCPs of small (near-linear) size \cite{BenSassonGHSV2006,BenSassonGHSV2005,MoshkovitzR2010-linearpcp,MoshkovitzR2010-2qpcp}. We do not pursue this direction in this paper, though it can be a subject of further study.)

The main result of this paper is a single analysis of low-degree testing that qualitatively subsumes all previous low-degree tests. We analyze the lines-point low-degree test and give an analysis showing it has error tolerance $1 - (\nicefrac{d}{q})^{\Omega(1)}$ when $q = O(d)$ (See \cref{thm:multivariate-low-degree-test-intro}). 

\subsection{Technical Contributions} 

Most analyses of low-degree tests follow the following paradigm: One first analyzes the low-degree test in the setting of $m = O(1)$, i.e., with a constant number of variables. And then a second step of analysis ``bootstraps'' the result from $O(1)$-variables to general $m$ variables. Our improvement follows the same paradigm and contributes to both steps. We explain our contribution to the two steps below by contrasting with the previous works. 

Previous techniques in the $O(1)$-variable setting come in two distinct flavors: The Raz-Safra analysis \cite{RazS1997} is very coding theoretic. For instance when testing 3-variate functions by picking planes uniformly, the analysis relies on the fact that two typical planes intersect in a line, and on this line the nearest polynomial is a codeword of a code with $1 - o(1)$ distance (if $d = o(q)$). This very high distance of the underlying code is critical to their analysis. The Raz-Safra analysis is thus very clean, but it simply can not work with the lines-point test where typical lines intersect in at most one point.

The $O(1)$-variable analyses for the line-point tests in \cite{PolishchukS1994} and in \cite{AroraS2003} are both very algebraic. Both rely on ideas used to decode Reed-Solomon codes, but need some heavy machinery to work with these. \cite{PolishchukS1994} in particular use properties of the derivatives of resultants to effect their analysis; and \cite{AroraS2003} uses effective Hilbert irreducibility to effect their analysis. While such use of heavy machinery is inevitable given the approach, the weakness in previous results comes from the black box use of the tools that they use. For instance \cite{PolishchukS1994} effectively only uses the fact that when $m=2$, the low-degree test effectively gives two directions in which the function looks like a low-degree polynomial in that direction, but is unable to use the fact that $\Omega(1)$-fraction of directions actually have this property. \cite{AroraS2003} also suffers from a black box use of Hilbert irreducibility. Our key contribution (and we elaborate more on this in \cref{sec:proof-overview}) is to open up this black box and adapt it to our setting. (Our proofs as a result are self-contained, and arguably simpler than those in \cite{AroraS2003}.)

Turning to the ``bootstrapping'' --- here there are roughly three previous works to compare against. The Raz-Safra bootstrapping \cite{RazS1997} turns out to be the weakest and shows that an error tolerance of $1 - \epsilon$ in the $O(1)$-variable setting can be converted to an error tolerance of $1 - O(m\eps)$ in the  $m$-variate setting. While this may be adequate in some settings, this is certainly not the right answer. In the low-error setting, Friedl and Sudan~\cite{FriedlS1995}, essentially building on Rubinfeld and Sudan~\cite{RubinfeldS1996}, show that an error upper bound of $\epsilon$ in the $2$-variable setting can be converted to an error bound of $\epsilon/C$ (for some universal constant $C$) in the $m$-variate setting; but their proof is very highly tailored to the unique-decoding setting, i.e., the proofs use  the fact that in such settings there is at most one polynomial at $\epsilon$ distance from any given function. Arora and Sudan~\cite{AroraS2003} extend this analysis to the list-decoding setting but their proof is complex and finally is only able to reduce to the $3$-variable case. Our analysis shows that an error-tolerance of $1 - \epsilon$ in the $2$-variable setting implies a $1 - \epsilon^{\Omega(1)}$ error tolerance in the general $m$. 
While our analysis is in the same spirit as the previous works, is substantially cleaner and manages to cleanly reduce to the 2-variable case. We remark that our analysis is inspired by the clean local-list-decoder for multivariate polynomial codes due to Sudan, Trevisan and Vadhan~\cite{SudanTV2001} who in turn clean up a similar result from \cite{AroraS2003}, an idea that does not seem to have been explored much in the interim period.

We remark that both the \cite{FriedlS1995} analysis and the \cite{AroraS2003} analysis can be interpreted as exploiting some expansion properties of underlying high-dimensional expanders (HDXs). The \cite{FriedlS1995} HDX consists of three layers, the points in $\F_q^m$, the lines in $\F_q^m$ and some complicated 2-dimensional surfaces in $\F_q^m$. The \cite{AroraS2003} HDX is simply the Grasmannian (points, lines, planes and cubes), but now it has four layers. The cleanliness of our analysis is highlighted by the fact that we also work with the (affine) Grassmannian in $\F_q^m$, but now again with three layers (points, lines and planes). 

Finally we remark that we do not optimize the dependence between $q$ and the $\epsilon$ in the error-tolerance we obtain here. Recent works due to Bhangale, Dinur and Navon \cite{BhangaleDN2017} and Minzer and Zheng \cite{MinzerZ2023} explore this connection and obtain near-optimal dependence between $\eps$ and $q$, albeit at the cost of an even larger query test -- a ``cubes-point'' test. However, these results do not optimize the dependence on $d$. These results are proved by using deeper expansion properties of the 4-layered Grassmannian (points, lines, planes and cubes) in $\F_q^m$. 
If one obtains an optimal $\eps$-vs-$q$ tradeoff for the lines-point test, with $q = O(d)$ as in our proof, this might yield hardness of labelcover of the following type:
NP-hard to distinguish between labelcover instances with alphabet size $O(q^2)$ with perfect completeness from those that are at most $O(\nicefrac{(\log q)}{q})$-satisfible. The current best result, in this context, is due Siu on Chan \cite{Chan2016}, albeit with imperfect completeness. We believe these questions merit further study and view our work as a step towards understanding them.

\subsection{Our results}\label{sec:results}

To state our results which work in the high-error regime of the lines-point test, it will be best to first recall similar results in the low-error regime. We begin with the following lemma due to Polishchuk and Spielman \cite{PolishchukS1994}.

\begin{lemma}[Polishchuk-Spielman Lemma {\cite[Lemma~8]{PolishchukS1994}}, see also {\cite[Lemma 4.3]{BenSassonCIKS2020}}]\label{lem:pslemma}
  Let $\F$ be any field and let $A(x,y,z) \in \F[x,y, z]$ be a non-zero trivariate polynomial with $z$-degree at most 1 and $(1,0,d)$-degree and $(0,1,d)$-degree at most $D$. Furthermore, suppose there exist two sets $U, V \subseteq \F$ such that (1) for every $u \in U$, there exists a degree-$d$ univariate polynomial $C_u(y)$ such that $A(u,y,C_u(y))\equiv 0$ and similarly (2) for every $v \in V$, there exists a a degree-$d$ univariate polynomial $R_v(x)$ such that $A(x,v,R_v(x))\equiv 0$. 

  If the sets $U, V$ are of size at least $2D$ each then there exists a polynomial $Q(x,y)$ of individual degree at most $d$ in each variable such that 
  \[ 
    A(x,y,Q(x,y)) \equiv 0.
  \]  
\end{lemma}

\noindent
Friedl and Sudan \cite{FriedlS1995} used the above lemma to prove the following low-error lines-point theorem. 

\begin{restatable}[Multivariate low-error LDT {\cite{FriedlS1995}}]{theorem}{multivariatelowerrorLDT}
\label{thm:main-lowerror-ldt} 
  There is a constant $C$ large enough such that for every finite field $\F_q$, $m > 0$, degree $d$ satisfying $q > C \cdot d$, the following holds:

  \begin{quote}
    Suppose the points table $f \colon \F_q^m \to \F_q$ and the degree-$d$ lines oracle $P_\ell$ satisfy
    \[ \Pr_{\va,\ell}[f(\va)\neq P_\ell(\va)] \leq \delta,\]
    for some $0 < \delta < 0.01$, then there exists a $m$-variate degree-$d$ polynomial $Q$ such that $\delta(f,Q)\leq 4\delta$.
  \end{quote}
\end{restatable}

We are interested in proving high-error versions of the above results. A natural approach will be to extend the Polishchuk-Spielman \cref{lem:pslemma} to higher degrees in $z$. However, the natural generalization happens to be false. Consider the trivariate polynomial \[A(x,y,z):=(xz-C(y))\cdot (yz-R(x))\]where $C \in \F[y], R \in \F[x]$ are two polynomials of degree at most $d$. Clearly for each $u\neq 0$, $A(u,y,\nicefrac{C(y)}{u})\equiv 0$ and for each $v\neq 0$, $A(x,v,\nicefrac{R(x)}{v})\equiv 0$. Yet, there exists no $Q(x,y)$ such that $A(x,y,Q(x,y)) \equiv 0$. This counterexample exists as $z$-degree of $A$ is at least the number of parallel directions\footnote{There exist similar counterexamples for every $r \geq 2$ with $z$-degree being $r$ and a set of $r$-parallel directions. The above example with $r=2$ has 2 sets of parallel directions (lines parallel to the $x$-axis and those parallel to the $y$-axis).}. Our main lemma shows that if the number of parallel directions is considerably more than the $z$-degree, then such counterexamples do not exist. We find it more convenient to state our lemma for sets of lines passing through a point rather than sets of parallel lines. \footnote{A set of parallel lines in any direction can be viewed as a set of lines in different directions through a single point at infinity. In this sense, working with lines through a fixed point (a configuration that we refer to as a pencil), and sets of parallel lines are essentially equivalent. For our arguments, the former happens to be a bit more natural. } 

We mention a slightly informal statement of our main lemma. 
\begin{lemma}[Main technical lemma (Informal)]\label{lem:pencil-lemma-intro}
 
Let $\F$ be any finite field and $A(x,y,z) \in \F[x,y,z]$ be a non-zero trivariate polynomial with $z$-degree at most $d_z$ and $(1,1,d)$-weighted degree at most $D$ such that the characteristic of $\F$ is greater than $d_z$. 

If there is a set $B \subseteq \F^2$ that satisfies  (1) $|B| > 2d_zD|\F|$, and (2) for every $(\alpha, \beta) \in B$, there is a set $S_{(\alpha,\beta)} \subseteq \F^2$ of size greater than $(d_zD|\F|)$ such that for every $(u,v) \in S_{(\alpha,\beta)}$, there exists a univariate polynomial $P_{((\alpha,\beta),(u,v))}(t) \in \F[t]$ of degree at most $d$ for which  \[A\left(\alpha+tu,\beta+tv,P_{((\alpha,\beta),(u,v))}(t)\right) \equiv 0 \, ,\] 
then, there exists a polynomial $P(x,y) \in \F[x,y]$ of total degree at most $d$ such that 
\[
A(x,y,P(x,y)) \equiv 0 \, .
\]
\end{lemma}
\begin{remark*}
Even though \cref{lem:pencil-lemma-intro} is stated for trivariate polynomials here, the statement is true for multivariate polynomials as it is. Moreover, even though the underlying field $\F$ is assumed to be finite here and with large enough characteristic, a similar statement that suffices for our applications to low-degree testing is true  over all fields. We refer to \cref{sec: PS on pencil of lines} for details. 
\end{remark*}
Given this lemma, we can now prove the following high-error bivariate lines-point theorem. 

\begin{restatable}[Bivariate low-degree test]{theorem}{bivariateldtintro}
  \label{thm:bivariate-high-agreement-LDT}
  There exists a constant $\tau\in (0,1]$ such that for every finite field $\F_q$, $\epsilon_0 \in (0,1)$ and degree $d$ satisfying $\epsilon_0 > (d/q)^{\tau}$, the following holds. 

  Suppose the points table $f\colon\F_q^2\to \F_q$ and degree-$d$ lines oracle $P_{\ell}$ satisfy
  \[
  \Pr_{\va \in \F_q^2, \ell \ni \va}[f(\va) = P_{\ell}(\va)] \geq \epsilon \, , 
  \]
  for some constant $\epsilon > \epsilon_0$. Then, there is a bivariate degree $d$ polynomial $Q$ such that 
      \[
        \Pr_{\va\in \F_q^2}[f(\va) = Q(\va)] \geq \epsilon - \epsilon_0.
      \]
\end{restatable}

An added advantage of the above theorem is that it also works in the low-error regime. This gives a single proof (in the bivariate setting) that works both for both the low-error and high-error regimes. We can then bootstrap the bivariate low-error theorem using an argument similar to (but simpler than) \cite{FriedlS1995} to give an alternate proof of multivariate low-error Friedl-Sudan \cref{thm:main-lowerror-ldt}. Finally, we bootstrap the high-error version of the above theorem to yield the following multivariate high-error result. 

\begin{theorem}[Multivariate low-degree test]\label{thm:multivariate-low-degree-test-intro}
There exists a constant $\tau\in (0,1]$ such that for every finite field $\F_q$, $\epsilon_0 \in (0,1)$, $m$, and degree $d$ satisfying $\epsilon_0 > (d/q)^{\tau}$, the following holds. 

Suppose the points table $f\colon\F_q^m\to \F_q$ and degree-$d$ lines oracle $P_{\ell}$ satisfy
\[
\Pr_{\va \in \F_q^m, \ell \ni \va}[f(\va) = P_{\ell}(\va)] \geq \epsilon \, , 
\]
for some constant $\epsilon > \epsilon_0$. Then, there is a bivariate degree $d$ polynomial $Q$ such that 
    \[
      \Pr_{\va\in \F_q^m}[f(\va) = Q(\va)] \geq \epsilon - \epsilon_0.
    \]

\end{theorem}

We thus not only prove the multivariate low-error LDT \cref{thm:multivariate-low-degree-test-intro}, but also given an alternate (and arguably simpler) proof of the the multivariate low-error LDT \cref{thm:main-lowerror-ldt}. This yields a completely self-contained treatment of the lines-point test in both error regimes. 

\subsection{Proof overview} \label{sec:proof-overview}
We give a high-level overview of the main ideas in the proofs of our results. The overall structure of our argument proceeds as follows --- (1) Analyse the bivariate low-degree test, and (2) \emph{bootstrap} the argument to $m$-variate low-degree tests. 

All of our key algebraic ideas already make an appearance in the proof of the bivariate low-degree test, which is then \emph{lifted} to the multivariate setting using combinatorial techniques, in particular the expansion properties of the points-lines-planes affine Grassmannian.

\subsubsection{Bivariate low-degree tests} \label{sec:overview-bivariate-LDT}

Suppose are given a `points table' $f\colon \F_q^2 \to \F_q$ that passes the low-degree test with probability $\epsilon$. The argument for \cref{thm:bivariate-high-agreement-LDT} proceeds as follows:

\begin{enumerate}
  \item  \label{overview:step-explainer}
  \textbf{Find an `explainer':}
  We will show that there is a trivariate polynomial $A(x,y,z)$ of $(1,1,d)$-weighted degree at most $D$ that satisfies $A(a,b,f(a,b)) = 0$ for a structured set $S$ of  $\poly(\epsilon)$ fraction of the points $(a,b) \in \F_q^2$. 
  \item \label{overview:step-local-roots}
  \textbf{Show that $A$ has `low-degree roots' on many lines:}
  We then find a `good' point $(b_1,b_2) \in \F_q^2$ such that on \emph{many} lines $\ell_{\vb,\vu} = \setdef{(b_1 + tu_1, b_2 + tu_2)}{t\in \F_q}$ through this point, we have a degree $d$ polynomial $P_{\vb,\vu}(t)$ such that $A(b_1 + tu_1, b_2 + tu_2, P_{\vb,\vu}(t)) \equiv 0$. 
  \item \label{overview:step-global-roots}
  \textbf{Show that $A$ must have a global `low-degree root':}
  With the hypothesis from the previous step, we use \cref{lem:pencil-lemma-intro} to show that there must be a degree $d$ polynomial $Q(x,y)$ such that $A(x,y,Q(x,y)) \equiv 0$. 
  \item \label{overview:step-poly-agreement}
  \textbf{Show that $f$ must have non-trivial agreement with $Q(x,y)$}
  From the polynomial obtained in the previous step, and the structure of $S$, we show that $f$ and $Q$ must agree on $\Omega(\epsilon^4 \cdot q^2)$ points. 
\end{enumerate}

Once we have non-trivial agreement as in \cref{overview:step-poly-agreement}, standard reductions in the low-degree testing literature yields the stronger form as stated in \cref{thm:bivariate-high-agreement-LDT}.

\cref{overview:step-explainer,overview:step-local-roots} proceed along similar lines as in the results of Arora and Sudan~\cite{AroraS2003}, with some additional care to ensure that the $z$-degree of $A$ is function of just $\epsilon$ (and not $d$). With this additional care, the argument is able to interpolate such an `explainer' when the density of $S$ is $O(D/|\F|)$ as opposed to $\poly(D)/|\F|$ in \cite{AroraS2003}. The key technical step is \cref{overview:step-global-roots} which uses \cref{lem:pencil-lemma-intro}. This lemma, in spirit, can be thought of as an analog of the celebrated lemma of Polischuk \& Spielman \cite{PolishchukS1994} in the high error setting. \cref{overview:step-poly-agreement} is almost an immediate consequence of the previous steps.

\paragraph{Sketch of the proof of \cref{lem:pencil-lemma-intro}: } Since we have a `local root' for many lines, we first observe that many of the polynomials $P_{\vb,\vu}(t)$ must share the same constant term $\alpha$. Using the fact that $(b_1,b_2)$ is a `good' point, this ensures some non-degeneracy properties for the point $(b_1,b_2,\alpha)$ satisfying $A(b_1,b_2,\alpha) = 0$. We then use Newton Iteration to construct an ``approximate'' root $\Phi_k(x,y)$ such that $A(b_1 + x, b_2 + y, \Phi_k(x,y)) = 0 \bmod{\inangle{x,y}^{k+1}}$, for any choice of $k$. (That is, $\Phi_k(x,y)$ is a root if we are willing to ignore some high-degree terms.) Finally, by using some uniqueness properties guaranteed by Newton Iteration, we show that many of the $P_{\vb,\vu}$'s must infact be restrictions of $\Phi_d(x,y)$ to the respective lines. This allows us to eventually argue that $\Phi_d(x,y)$ must infact be an global root of $A(x,y,z)$. 

In a broader sense, the above proof opens up the use of Hilbert's Irreducibility Theorem in \cite{AroraS2003} and makes appropriate changes to ensure the better dependence in parameters. A more detailed discussion on the differences is provided below. 

\paragraph{Dealing with fields of small positive characteristic: }
In general, proofs obtained via derivative based techniques like Taylor expansion or Newton iteration suffer from technical issues when applied over fields of small characteristic. Among examples of such results are the results on polynomial factorization \cite{Kaltofen1989,KoppartySS2015}, results on list decoding of multiplicity codes \cite{GuruswamiW2013,Kopparty2014, Kopparty2015,BhandariHKS2023-mgrid} and hardness-randomness tradeoffs in algebraic complexity \cite{GuoKSS2022,  Andrews2020}. Intuitively, the issue stems from the fact over a field of small characteristic, a polynomial can depend on a variable, but its partial derivative with respect to this variable can still be identically zero, e.g. if the polynomial is a function of $x^p$ over a finite field of characteristic $p$. Our outline also suffers from these issues. However, we observe that these issues can be resolved and that the results extend to all finite fields  via one simple additional idea --- while interpolating an explainer, ensure that $A(x,y,z)$ has a non-zero partial derivative with respect to $z$ by just ignoring all monomials whose exponent in $z$ is divisible by $p$. As it turns out, this simple modification still allows us to proceed with the rest of the argument without hardly any change to the parameters involved.

\paragraph{Technical differences with the analysis of Arora-Sudan: }
While our proof in the bivariate case is conceptually similar to that of Arora-Sudan \cite{AroraS2003}, there are some differences that lead to better quantitative bounds. The first technical difference is that in the interpolation step of the analysis, we interpolate a trivariate polynomial $A(x, y, z)$ of $(1, 1, d)$-degree $D$, whereas Arora-Sudan work with polynomials of total degree $D$. As a consequence, the $z$-degree of $A$ can be bounded by at most $D/d$ and not just $D$. In particular, when $D = O(d)$, the $z$-degree of $A$ happens to be a constant for us, whereas this is not the case in \cite{AroraS2003}.

The second technical difference is in the statement of a lemma analogous to \cref{lem:pencil-lemma-intro} that Arora \& Sudan prove. In their lemma (Lemma 39 in \cite{AroraS2003}), they construct a polynomial of degree $O(D^3)$ that essentially detects whether for $u, v \in \F$, the restriction $A(ut, vt, z)$ of $A(x, y, z)$ has a factor that is linear in $z$. More precisely, if $A(x,y,z)$ does not have a factor that is linear in $z$, then they construct a non-zero polynomial $Q(x,y)$ of degree at most $O(D^3)$ such that if $Q(u,v)$ is non-zero, then  $A(ut, vt, z)$ \emph{does not} have a factor that is linear in $z$. Two points to note here are that the degree of $Q$ is polynomially larger than $D$, and that this statement talks about factors that are linear in $z$, but is not sensitive to the degree of the such a factor in $x,y$ (which could be as large as $D$). As it turns out, the analysis of the bivariate low-degree test only ever cares about factors that are linear in $z$ and have degree at most $d$ in $x,y$. Being aware of $(x,y)$-degree of these factors of interest, and additionally considering polynomials of $(1,1,d)$-degree at most $D$ lets us construct an analog of the polynomial $Q$ above whose degree is just $D$. This reduction in degree of $Q$ then naturally translates to a reduction in the field size requirement in the overall analysis.

\subsubsection{Bootstrapping to higher dimensions}\label{sec:proof-overview-higher-dim}

Having proved the low-degree test theorem for dimension $m=2$, we now need to bootstrap it to higher dimensions. There are several possible ways to do this. One potential route is the following. Given the lines-point bivariate LDT (i.e., $m=2$), we obtain a degree-$d$ planes oracle that passes the planes-point test with non-trivial probability and we can then  use the Raz-Safra bootstrapping analysis \cite{RazS1997} to bootstrap to arbitrary dimensions. This however causes $\eps_0$ (in \cref{thm:multivariate-low-degree-test-intro}) to be at least $\Omega(m \cdot \left(\nicefrac{d}{q}\right)^\tau)$. An alternate route to get around this linear dependence in $m$ is the following route. Bootstrap using the Raz-Safra analysis to dimension $m=3$ and then use the Bhangale-Dinur-Navon bootstrapping argument \cite{BhangaleDN2017} to bootstrap to arbitrary dimensions. The \cite{BhangaleDN2017} analysis (as written in their paper) requires $\eps_0$ to be at least $\Omega(\left(\nicefrac{d^8}{q}\right)^\tau)$. While this avoids the $m$-dependence, it requires the field size $q$ to be at least $d^8$. The \cite{BhangaleDN2017} analysis can be tighted to yield a $q = O(d)$ dependence by using the Friedl-Sudan analysis \cite{FriedlS1995} instead of the Rubinfeld-Sudan analysis \cite{RubinfeldS1996} which they use. To avoid these issues and give a self-contained proof, we give a direct bootstrapping argument from $m=2$ to arbitrary dimensions (inspired of course by bootstrapping analyses of \cite{FriedlS1995,AroraS2003,RazS1997,BhangaleDN2017}). However, since we intend to lift from the bivariate LDT (that is, the line-point test in a plane), some effort is required to make the broad ideas of bootstrapping work in our regime. As mentioned earlier in the introduction, this direct bootstrapping is inspired by the clean local-list-decoder for multivariate polynomial codes due to Sudan, Trevisan \& Vadhan \cite{SudanTV2001}.

\medskip

As a warmup, let us first explain the bootstrapping for the low-error regime, a la Friedl-Sudan \cite{FriedlS1995}. In this setting, we are given `points table' $f\colon \F_q^m \rightarrow \F_q$ and the best degree-$d$ lines oracle $P^{(f,d)}$ that fails the LDT with probability at most $\delta$, we wish to show that there is a degree-$d$ polynomial $Q(x_1,\ldots, x_m)$ that is $O(\delta)$-close to $f$. Following Friedl-Sudan, we define a corrected function $f_\corr$ as follows: for any point $\vy \in \F_q^m$, $f_\corr(\vy)$ is the most popular value of $P^{(f,d)}_\ell(\vy)$ among all lines $\ell$ passing through $\vy$ where $P^{(f,d)}_\ell$ is the best-fit degree-$d$ polynomial $P^{(f,d)}_\ell$ agreeing with $f$ on the line $\ell$ (breaking ties arbitrarily). Let $\delta_f$ be the rejection probability of low-degree test when run on the points table $f$ and the best-fit lines oracle $P^{(f,d)}_\ell$ for $f$. We will show that as long as $\delta_f \leq \delta_0$ for some fixed constant (dependent on $\nicefrac{d}{q}$), the rejection probability of the corrected function $f_\corr$, names $\delta_{f_\corr}$ is significantly smaller than that of $f$. More precisely, $\delta_{f_\corr} \leq \delta_f/2$. Once we have shown this, we can repeat the self-correction procedure several times to eventually arrive at a corrected function $f_*$ such that $\delta_{f_*}=0$, in which case it is a low-degree polynomial. So, it suffices to show that the corrected function $f_\corr$ passes the low-degree test with significantly better probabibility than the original function $f$. To prove this, it suffices for us to show that for a random point $\vx$, the probability that two random lines $\ell, \ell'$ passing through $\vx$ satisfy $P^{(f,d)}_\ell(\vx)= P^{(f,d)}_{\ell'}(\vx)$. For any such triple $(\vx,\ell, \ell')$, consider the plane $\pi$ containing $\ell$ and $\ell'$. If the function $f$ restricted to this plane $\pi$ passes the low-degree with high probability, then the bivariate low-error LDT theorem states that restricted to this plane, the function $f$, mostly behaves like a degree-$d$ polynomial $Q_\pi$ and both $P^{(f,d)}_\ell$ and $P^{(f,d)}_{\ell'}$ are in fact restrictions of $Q_\pi$ and hence equal to each other on the point $\vx$. The the above arguments work only on average and not for every triple $(\vx,\ell,\ell')$. Nevertheless, we show that this suffices to bootstrap to higher dimensions $m$ and here we use the expansion properties of the points-lines-planes affine Grassmannian. This proof is adapted from (and arguably simpler than) the corresponding bootstrapping proof due to Friedl and Sudan \cite{FriedlS1995}, who use a different high-dimensional expander consisting of points-lines-surfaces. 

We now turn to bootstrapping in the high-error regime. Given a `points table' $f:\F_q^m \rightarrow \F_q$ that passes the LDT with probability $\epsilon$, we wish to show that there is a degree $d$ polynomial $Q(x_1,\ldots, x_m)$ that agrees with $f$ on $\poly(\epsilon)$ places (proceeding from a `weak-agreement' statement to the statement in \cref{thm:multivariate-low-degree-test-intro} again follows from standard reductions). As in the low-error case, we would like to define a corrected function such that the corrected function passes the low-degree test with significantly better probability than the original function. The primary issue here is that there could be \emph{many} candidate $Q$'s that have $\poly(\epsilon)$ agreement and hence several different choices for the corrected function. We use an additional advice (a random point $\vx$ and the value of the function $f$ at the point $\vx$, namely $f(\vx)$) to disambiguate among the several different choices. One such corrected function $g^{(\vx)}_\corr\colon \F^m \to \F$ is as follows.  For any $\vy \in \F_q^m$, $g^{(\vx)}_\corr(\vy)$ is the most popular value of $Q^{(\vx)}_\pi(\vy)$ among all planes $\pi$ containing both $\vx$ and $\vy$ such that $Q_\pi^{(\vx)}$ is the unique degree-$d$ polynomial that ``explains'' the plane $\pi$ and furthermore $Q^{(\vx)}_\pi(\vx)=f(\vx)$. Why does such a ``explaining'' polynomial $Q_\pi$ exist? If the restriction of the function $f$ to the plane $\pi$ passes the low-degree test with probability $\eps$ (which happens if the plane $\pi$ is random), then the bivariate LDT theorem states that there exists such a polynomial $Q_\pi$. This is precisely the bootstrapping argument of Arora and Sudan \cite{AroraS2003}. However, then to show that this corrected function $g^{(\vx)}_\corr$ passes the low-degree test with significantly better probability than $f$, one needs to consider a cube and this is why the \cite{AroraS2003} bootstrapping argument required as base case both the $m=2$ and $m=3$ cases. To get around this dependence on the $m=3$ case, we define an alternate correction function $f^{(\vx)}_\corr$ (also disambiguated using the advice $(\vx,f(\vx))$), that we describe informally here:
\begin{quote}
  Pick a random $\vx \in \F_q^m$ to use for constructing the correction. Define $f_\corr^{(\vx)}(\vx) = f(\vx)$. 

  For every other point $\vy \in \F_q^m$, consider the line $\ell_{\vx,\vy}$ passing through $\vx$ and $\vy$. Find the list of all polynomials $P(t)$ that agree with $f$ on this line on at least $\poly(\epsilon)$-fraction of places. Among those, if there is a \emph{unique} $P$ that takes value $f(\vx)$ on $\vx$, set $f_\corr^{(\vx)}(\vy) = P(\vy)$ and set it to $\bot$ (or a random value in $\F_q$) otherwise. 
\end{quote}
In other words, the line joining $\vx,\vy$ is being used to `correct' the value at $y$ but we are disambiguating between the potential possibilities on this line by the value those polynomials take at the point $\vx$. 

Having defined this corrected function $f^{(\vx)}_\corr$, it is not hard to show for with non-trivial probability over the choice of the advice point $\vx$, we have that the corrected function $f^{(\vx)}_\corr$ has non-trivial $\poly(\eps)$ agreement with the function $f$. 

We now try to understand the behaviour of the corrected function on a random line $\ell$. Consider the plane $\pi$ containing the advice point $\vx$ and the line $\ell$. Using the bivariate LDT theorem, we can show that for a random plane $\pi$ a non-trivial fraction of the points $\vx$ in $\pi$, there exists a degree-$d$ polynomial $Q_\pi^{(\vx)}$ that has $\poly(\eps)$-agreement with $f$ on $\pi$ and $Q^{(\vx)}_\pi(\vx)=f(\vx)$. We use this coupled with the expansion properties of the lines-points-planes affine Grassmannian to show that for a random $\vx$ and a random line $\ell$, the corrected function $f^{(\vx)}_\corr$ (using advice $\vx$) agrees with the polynomial $Q^{(\vx)}_\pi$ overwhelmingly on the line where $\pi$ is the plane containing $\vx$ and $\ell$. This step requires a careful analysis by using spectral properties of various natural subgraphs of the Affine Grassmanian. We thus, have, $\Pr_{\vx,\ell}[ \delta(f^{(\vx)}_\corr|_\ell, Q^{(\vx)}_\pi(x,\ell)|_\ell) \leq \gamma]\geq 1-2\gamma$. This implies that for a random $\vx$, the corrected function $f^{(\vx)}_\corr$ passes the low-degree test overwhelmingly (in fact with probability $1-3\gamma$). Hence, by \cref{thm:main-lowerror-ldt}, there exists a degree-$d$ $m$-variate polynomial $Q^{(x)}$ such that $\delta(f^{(\vx)}_\corr,Q^{(\vx)}) \leq O(\gamma)$. Combining this with the fact that $f^{(\vx)}_\corr$ has $\poly(\eps)$ agreement with the function $f$, we obtain that $f$ has $\poly(\eps)-O(\gamma)=\poly(\eps)$ agreement with the polynomial $Q^{(\vx)}$ provided $\gamma \ll \poly(\eps)$. This completes the overview of the bootstrapping argument. 

\subsection{Organisation of the paper}

We begin with notation and preliminaries in \cref{sec:prelims} and then proof the main technical lemma (\cref{lem:pencil-lemma-intro}) in \cref{sec: PS on pencil of lines}. We then proceed to the analysis of the bivariate low-degree tests in \cref{sec:ldt-high-error}. Finally, we bootstrap the bivariate low-degree tests to $m$-variate low-degree tests in \cref{sec:bootstrapping}.

\section{Preliminaries}
\label{sec:prelims}

\paragraph{Notation:}

\begin{enumerate}\itemsep0pt
\item For a polynomial $f(x_1,\ldots, x_m)$ and a vector $\va = (a_1,\ldots, a_m) \in \N^m$, we use the term $\ve$-weighted degree to denote the largest value of $a_1 e_1 + \cdots + a_m e_m$ among monomials $x_1^{e_1}\cdots x_m^{e_m}$ in the support of $f$. 
\item For two functions $f,g:\F^m \to \F$, we define $\agree(f,g) = \Pr_{\va\in\F^m}[f(\va) = g(\va)]$, the fraction of points that the tables $f$ and $g$ agree on. Similarly, we will use $\delta(f,g)$ to denote $\Pr_{\va\in\F^m}[f(\va)\neq g(\va)]$, the fraction of points that the tables $f$ and $g$ disagree on.
\end{enumerate}

\begin{lemma}[Averaging argument] 
  \label{lem:averaging}
  Suppose $x_1,\ldots, x_n \in [0,1]$ such that $\sum x_i \geq \mu n$. Then,
  \begin{enumerate}
  \item \label{lem:averaging-item1} if $S = \setdef{i\in [n] }{x_i \geq \sfrac{\mu}{2}}$, then $|S| \geq \sfrac{\mu n}{2}$, and
  \item \label{lem:averaging-item2} $\sum_{i\in S} x_i \geq \sfrac{\mu n}{2}$. 
  \end{enumerate}
\end{lemma}
\begin{proof}
  For \cref{lem:averaging-item1}, note that
  \begin{align*}
    \mu n \leq \sum_i x_i & \leq \sfrac{\mu}{2} \cdot (n - |S|) + |S|\cdot 1 \leq \sfrac{\mu n}{2} + |S|\\
    \implies \sfrac{\mu n}{2} & \leq |S|.
  \end{align*}

  For \cref{lem:averaging-item2},
  \begin{align*}
    \sum_{i\in S} x_i & \geq \mu n - \sum_{i\notin S} x_i \geq \mu n - (n - |S|) \cdot \sfrac{\mu}{2}\\
    & \geq \mu n - (1 - \sfrac{\mu}{2}) \cdot \sfrac{\mu}{2} \cdot n \geq \sfrac{\mu n}{2}.\qedhere
  \end{align*}
\end{proof}

\subsection*{Polynomial identity lemma}

\begin{lemma}[\cite{Ore1922,DemilloL1978,Schwartz1980,Zippel1979}]\label{lem: SZ}
Let $P(\vx)$ be a non-zero polynomial of total degree at most $D$ with coefficients over a field $\F$ and let $S$ be an arbitrary subset of $\F$. Then, the number of zeroes of $P$ on the product set $S \times S \times \cdots \times S$ is at most $D|S|^{m-1}$. 
\end{lemma}

\subsection*{Hasse Derivatives and Properties}
Throughout the paper, we use the notion of Hasse derivatives of polynomials and some of their basic properties. We start with the definition. 
\begin{definition}\label{defn:hasse-derivatives}
Let $\F$ be any field and $A(\vx) \in \F[\vx]$ be an $m$ variate polynomial with coefficients in $\F$. Then, the Hasse derivative of $A(\vx)$ with respect to a monomial $\vx^{\ve}$, denoted by $\hpartial_{\vx^{\ve}}(A)$, is defined as the coefficient of the monomial $\vz^{\ve}$ when viewing $A(\vs + \vz)$ as a polynomial in $\vz$ variables with coefficients in the ring $\F[\vx]$.  
\end{definition}
Hasse derivatives are an extremely useful and recurrent theme in computer science with a variety of applications. For our proofs in the paper, we rely on the following properties. 
\begin{proposition}\label{prop:hasse-derivatives-properties}
Let $\F$ be any field. Then, the following are true. 
\begin{itemize}
\item For every $d \in \N$, $\hpartial_x(x^d) = dx^{d-1}$. 
\item For every pair of polynomials $A, B \in \F[x]$, $\hpartial(AB) = \hpartial(A)\cdot B + B \cdot \hpartial(A)$. 
\item For any polynomial $A \in \F[x]$ of degree at least one in $x$, and a finite field $\F$ of characteristic $p$, $\hpartial_x(A)$ is identically zero if and only if there exists a polynomial $\tilde{A} \in \F[x]$ such that ${\tilde A}^{p} = A$.

\end{itemize}
\end{proposition}

\subsection*{Discriminant}
\begin{definition}[Resultant and Discriminant]\label{defn:resultant}
Let $\F$ be any field and $A, B \in \F[x]$ be non-zero univariate polynomials in $x$ of degree $a, b$ respectively. Let $A_0, A_1, \ldots, A_a, B_0, B_1, \ldots, B_b$ be elements of $\F$ such that $A(x) = \sum_{i = 0}^a A_ix^i$ and $B(x) = \sum_{j = 0}^b B_j x^j$. Then, the \emph{Sylvester} matrix of $A$ and $B$ defined as follows. 
\[\emph{Sylvester}(A, B) = 
\begin{bmatrix}
        A_0  & A_1    & \dots  &        & A_{a} &        &       \\
            & \ddots & \ddots &        & \ddots & \ddots &       \\
            &        & A_0    & A_1    &        &  \dots & A_{a} \\
      B_{0}  & \dots  &        & B_{b}  &        &        &       \\
            &B_{0}   & \dots  &        & B_{b}  &        &       \\
            &        & \ddots & \ddots &        & \ddots &       \\
            &        &        & B_{0}  & \dots &        &B_{b} 
    \end{bmatrix}
\]
And, the resultant of $A$ and $B$ is defined as the determinant of the \emph{Sylvester} matrix of $A$ and $B$. Moreover, if $B = \hpartial_x(A)$, then the resultant of $A$ and $B$ is referred to as the \emph{Discriminant} of $A$ and $B$ and denoted by $\disc_x(A, B)$.  
\end{definition}

\begin{lemma}\label{lem:resultant-vs-gcd}
Let $\F$ be any field and $A, B \in \F[x]$ be non-zero univariate polynomials in $x$. Then, $A, B$ have a non-trivial GCD (the degree of GCD is at least one in $x$) if and only if their resultant is zero.   
\end{lemma}
\begin{definition}
Let $\F$ be any field. A polynomial $A(x) \in \F[x]$ is said to be square-free if there does not exist a polynomial $B \in \F$ of degree at least one such that $B^2$ divides $A$. 
 \end{definition}

\begin{lemma}\label{lem:discriminant-vs-gcd}
Let $\F$ be any field and $A \in \F[x]$ be a non-zero univariate polynomial such that $\hpartial_x(A)$ is non-zero. Then, $A$ is square-free if and only if the discriminant of $A$ is non-zero.
\end{lemma}

\subsection{Structure of minimal interpolating polynomials}
In this section, we prove the following simple lemma that will be crucial to our analysis of the bivariate test. The lemma essentially lets us assume some structural properties on the interpolating polynomials without loss of generality, and these properties turn out to be important for our analysis, especially when we are working over fields of small characteristic. 

\begin{lemma}
  \label{lem:minimal-vanishing-polynomial-is-squarefree}
  Let $\F$ be any field and $m \in \N$ be an integer, and let $S \subseteq \F^m$ be a set of points. Consider the space of polynomial $\mathcal{Q}_{S,x_1} = \setdef{Q(x_1,\ldots, x_m)}{Q(\va) = 0\text{ for all $\va \in S$ and $\hpartial_{x_1}(Q)\neq 0$}}$. For any vector $\vw \in \N^{m}$, the polynomial $Q \in \mathcal{Q}_{S,x_1}$ of lowest $\vw$-weighted degree satisfies $\disc_{x_1}(Q) \neq 0$ and is hence square-free.
\end{lemma}
\begin{proof}
  Let $Q \in \mathcal{Q}_{S,x_1}$ be the polynomial of smallest $\vw$-weighted degree. On the contrary, let us assume that $Q$ is \emph{not} square-free and say $Q = P^\ell \cdot R$ where $P$ is irreducible with $\ell \geq 2$ and $\gcd(P,R) = 1$. 

  Consider the polynomial $\tilde{Q} = P \cdot R$. Clearly $\tilde{Q}(\va) = 0$ for all $\va\in S$ as well, and if we can argue that $\hpartial_{x_0}(\tilde{Q}) \neq 0$, we have $\tilde{Q} \in \mathcal{Q}_{S,x_1}$ thus yielding a contradiction to the minimality of $Q$. 
  \begin{align*}
    \hpartial_{x_1}(Q) & = P^\ell \cdot \hpartial_{x_1}(R) + \ell P^{\ell-1} \cdot R \cdot \hpartial_{x_0}(P).\\
    \hpartial_{x_1}(\tilde{Q}) &= P \cdot \hpartial_{x_1}(R) + R \cdot \hpartial_{x_1}(P).
  \end{align*}
  Since we have $\hpartial_{x_1}(Q) \neq 0$, we have that at least one of $\hpartial_{x_1}(P), \hpartial_{x_1}(R)$ is non-zero. 
  If exactly one of $\hpartial_{x_1}(P), \hpartial_{x_1}(R)$ is non-zero, we immediately have that $\hpartial_{x_1}(\tilde{Q}) \neq 0$. Otherwise, since $P$ is irreducible and $\gcd(P,R) = 1$, we have that $P \nmid R \cdot \hpartial_{x_1}(P)$. Therefore, $P\nmid \hpartial_{x_1}(\tilde{Q})$ and this in particular forces $\hpartial_{x_1}(Q) \neq 0$ yielding our required contradiction. Thus $Q$ must have been square-free.
\end{proof}

\subsection{Power series roots}

The following lemma is a standard application of Newton iteration in multivariate polynomial rings and is an important ingredient of many of the multivariate factorization algorithms. We refer to \cite{Burgisser} for a proof.  
\begin{lemma}[Newton iteration for roots]
  \label{lem: taylor}
  Let $\F$ be a field and let $A(\vx, z) \in \F[\vx, z]$ be an $(m+1)$-variate polynomial and let $\alpha \in \F$ be a zero of multiplicity one of $A(\vzero, z)$, i.e., $A(\vzero, \alpha) = 0$ and $\hpartial_z(A(\vx, z))(\vzero, \alpha) = \hpartial_z(A(\vzero, z))(\alpha) \neq 0$. Then, for every $k \geq 1$, there is a polynomial $\Phi_k(\vx)$ such that the following are true.
  \begin{itemize}
  \item $A(\vx, \Phi_{k}(\vx))  \equiv 0 \mod \langle x_1, x_2, \ldots, x_m \rangle^{k+1}$, 
  \item $\alpha = \Phi_{k}(\vzero)  \equiv \Phi_{k}(\vx) \mod \langle x_1, x_2, \ldots, x_m \rangle$, 
  \item $\deg(\Phi_k) \leq k$.
  \end{itemize}
  Moreover, there is a unique polynomial that satisfies all these three properties. 
\end{lemma}

\section{Low-degree roots on restrictions to global low-degree roots}\label{sec: PS on pencil of lines}

We prove the following technical lemma (a more formal statement of \cref{lem:pencil-lemma-intro}) that is the main technical statement driving the analysis of our bivariate low-degree test. The lemma is true more generally (in higher dimensions) and we state and prove it in that way. 
\begin{lemma}\label{lem: decoding on good pencils}
  Let $\F$ be any field and let $A(\vx, z) \in \F[\vx, z]$ be a non-zero $(m+1)$-variate polynomial with $z$-degree at most $d_z$, $(1, 1, \ldots, 1, d)$ degree at most $D$ and let $\vb \in \F^m$ be a point such that the univariate polynomial $A(\vb, z) \in \F[z]$ has no repeated roots. Let $S \subseteq \F^m$ be a set of directions such that for every $\vu \in S$, there exists a univariate polynomial $P_{\vb, \vu}(t) \in \F[t]$ of degree at most $d$ that satisfies the identity 
  \[
  A(\vb + t\vu, P_{\vb, \vu}(t)) \equiv 0 \, .
  \]
  If $\left(|S| > d_zD|\F|^{m-1}\right)$, then there exists a polynomial $P(\vx) \in \F[\vx]$ of total degree at most $d$ such that 
  \[
  A(\vx, P(\vx)) \equiv 0 \, .
  \]
  
  Moreover, there is a set $S'\subseteq S$ of size at least $|S|/d_z$ such that for all $\vu \in S'$,  $P(\vx)$ when restricted to the line in direction $\vu$ through the point $\vb$ equals $P_{\vb, \vu}(t)$. 
\end{lemma}

\begin{proof}
  In its essence, the lemma above states that if for many different lines $L_{\vb, \vu}$ through a point $\vb$, the restriction $A(L_{\vb, \vu}, z) \in \F[t][z]$ of the polynomial $A(\vx, z)$ has a low-degree polynomial root $P_{\vb, \vu}(t)$, then the original unrestricted polynomial $A(\vx, z) \in \F[\vx][z]$ has a low-degree root $P(\vx)$. Moreover, we are given that the point $\vb$  is non-degenerate in the sense that the univariate polynomial $A(\vb, z)$ is square free. Qualitatively, the lemma is of flavor similar to the Hilbert's irreducibility theorem and the proof will be along similar lines but technically simpler and as it eventually turns out, quantitatively better for our eventually applications. \\

  For ease of notation, we will assume that $\vb = \vzero$ (by translating $A$ if necessary) and use $P_{\vu}$ to denote $P_{\vb, \vu}$. Let $\inangle{\vx}$ denote the ideal $\inangle{x_1,\ldots, x_m}$. 

  From the hypothesis of the lemma, we know that for every $\vu \in S$, there is a polynomial $P_{\vzero, \vu}(t)$ of degree at most $d$ such that $A(t\vu, P_{\vu}(t)) \equiv 0$, which implies that $A(\vzero,P_{\vu}(0)) = 0$. Therefore, if $\set{\alpha_1,\ldots, \alpha_\ell}$ is the multiset of roots over $\F$ for the polynomial $A(0,z)$, we must have that $P_\vu(0) = \alpha_i$ for some $i \in [\ell]$. As the $\deg(A(0,z))\leq d_z$, we have that $\ell \leq d_z$. Hence, there must exist some $i$ such that $\alpha_i = P_\vu(0)$ for at least $D \abs{\F}^{m-1}$ many $\vu \in S$. Let us refer to this $\alpha_i$ as just $\alpha$, and let $R = \setdef{\vu \in S}{P_\vu(0) = \alpha}$. Therefore, we have that $P_\vu(0) = \alpha$ for all $\vu \in R$, and also 
  \[
  A(t\vu, P_\vu(t)) = 0 \implies A(t\vu, P_\vu(t)) = 0 \bmod{\inangle{\vx}^{k+1}} \text{ for all $k \geq 0$}.
  \]

  On the other hand, since we are given that $A(\vzero,z)$ is square-free, the roots of this polynomial are distinct and hence $\hpartial_z(A(\vzero,z)) \neq 0$. Thus, by \cref{lem: taylor}, there exist polynomials $\Phi_k(\vx)$ for every $k \geq 0$ such that 
  \begin{align*}
    \Phi_k(\vzero) & = \alpha,\\
    A(\vx, \Phi_k(\vx)) & = 0 \bmod{\inangle{\vx}^{k+1}}.
  \end{align*}
  By making the substitution $x_i \mapsto t x_i$, the above equation transforms to $A(t\vx, \Phi_k(t\vx)) = 0 \bmod{t^k}$. In particular, we have
  \[
  A(t\vu, \Phi_k(t\vu)) = 0\bmod{t^k} \quad\text{for all $\vu\in R$}.
  \]
  Fix an arbitrary $\vu\in R$ and consider the polynomial $A'_\vu(t,z) := A(t\vu,z) \in \F[t,z]$. We know that $A'_\vu(0,\alpha) = 0$ is square-free and $\hpartial_z(A'_\vu(t,z))(0,\alpha) \neq 0$. Thus, by applying \cref{lem: taylor} on $A'_\vu(t,z)$ to lift from $\alpha$, there ought to be a unique polynomial $\Psi_\vu(t)$ of degree at most $d$ such that $\Psi_\vu(0) = \alpha$ and 
  \[
    A'_\vu(t,\Psi_{\vu}(t)) = A(t\vu, \Psi_\vu(t))  = 0 \bmod{t^{d+1}}.
  \]
  We already have two such candidates for $\Psi_\vu(t)$, namely the polynomials $P_\vu(t)$ and $\Phi_d(t\vu)$. Thus, by the uniqueness asserted by \cref{lem: taylor}, we have that 
  \[
  P_\vu(t)  = \Phi_d(t\vu) \text{ for all $\vu \in R$}. 
  \]

  To finish the proof, we wish to argue that $\Phi_d(\vx)$ is a true root of $A(\vx,z)$, i.e. it satisfies $A(\vx, \Phi_d(\vx)) = 0$ (even without $\bmod{\inangle{\vx}^{d+1}}$). To this end, consider the polynomial $B(\vx) = A(\vx, \Phi_d(\vx))$. Since $\deg(\Phi_d(\vx)) \leq d$ and $A$ has $(1,\ldots,1,d)$-weighted degree at most $D$, we have $\deg(B(\vx)) \leq D$. 
  
  For every $\vu \in R$, we have $\Phi_d(t\vu) = P_\vu(t)$ and hence $B(t\vu) = A(\vu, \Phi_d(t\vu)) = A(t\vu, P_\vu(t)) = 0$ as a polynomial in $t$, and thus $B(\vu) = 0$. This implies that the $m$-variate polynomial $B$ of degree at most $D$ has more than $D \cdot \abs{\F}^{m-1}$ zeros in $\F^m$. By the Polynomial Identity Lemma (\cref{lem: SZ}), conclude that $B(\vx) = A(\vx, \Phi_d(\vx))$ is the zero polynomial. 
\end{proof}

\section{The bivariate low-degree test}\label{sec:ldt-high-error}

In this section, prove \cref{thm:bivariate-high-agreement-LDT} which we recall for convenience.

\bivariateldtintro*

\begin{remark}[Version with a points table and a lines table] \label{rem:lines-table-version}
  Many results in the low-degree testing literature often consider a variant where we are given a `points table' $f:\F^m \rightarrow \F$, and a `lines table' $P:\mathcal{L}^{(m)} \rightarrow \F[x]^{\leq d}$ (where $\mathcal{L}^{(m)}$ denotes the set of lines in $\F^m$) that assigns a polynomial of degree at most $d$ for every line. The version above is a special case where $P$ is the \emph{canonical lines table} that assigns the best-fit degree $d$ polynomial on each line. The discussion in the section carries over in a straightforward manner to the more general setting of points and lines table, and we deal with the special case purely to avoid some additional notational clutter.
\end{remark}

Before we proceed with the proof of the above theorem, we will fix some notation. 
Let $f: \F^2 \rightarrow \F$ be provided as an oracle. Throughout this section, we will be using $\LDT_d$ to refer to the line-point test: 
\begin{itemize}\itemsep0pt
  \item Pick $x \in_R \F^2$ and a random line $\ell \ni x$. 
  \item Query the oracle on all points on the line $\ell$ and let $P_\ell$ be the best-fit degree $d$ polynomial.
  \item Accept if $P_\ell(x) = f(x)$. 
\end{itemize}
We will use $\epsilon_x$ to refer to $\Pr_{\ell \ni x}[P_\ell(x) = f(x)]$. \\

\noindent
To begin with, we will prove the following ``weaker'' statement for the low-degree test. 

\begin{theorem}
  \label{thm:LDT-low-agreement-poly-agreement}
  Suppose $\epsilon > 0$, and $d$ is a positive integer. Let $\F$ is a finite field with $q$ elements $q > C \cdot \sfrac{d}{\epsilon^7}$ for a large enough constant $C$. Suppose $f: \F^2 \rightarrow \F$ that passes $\LDT_d$ with probability at least $\epsilon$, that is
  \[
  \Pr_{a, \ell \ni a }[f(a) = P_\ell(a)] \geq \epsilon.
  \]
  Then, there is a polynomial $Q(x,y)$ of degree at most $d$ such that 
  \[
    \Pr_{a\in \F^2}[f(a) = Q(a)] \geq \Omega(\epsilon^4).
  \]
\end{theorem}

Although the above theorem appears to yield a weaker agreement than claimed by \cref{thm:bivariate-high-agreement-LDT}, it would turn out that \cref{thm:LDT-low-agreement-poly-agreement} yields \cref{thm:bivariate-high-agreement-LDT} via standard reductions that we elaborate on towards the end of this section. Thus, we proceed towards proving \cref{thm:LDT-low-agreement-poly-agreement}. \\

As a first step, we show the existence of a low-degree trivariate polynomial that \emph{explains} that the given function passes the line point test with a constant probability. 

\begin{restatable}[Interpolation for LDT]{theorem}{interpolationstep}
  \label{lem:interpolation-step}
  There exist constant $c_1, c_2 \in \N$ such that  for every $\epsilon \in (0, 1]$, any $d \in \N$ and finite field $\F$ of size $q$ and characteristic $p$ with $q > c_1\cdot d/\epsilon^{c_2}$, the following is true.

  Let $f:\F^2 \rightarrow \F$ be an oracle that passes $\LDT_d$ with probability $\epsilon$. Then, there is a non-zero polynomial ${A}(x,y,z)$ and a subset $S \subseteq \F^2$ such that 
  \begin{itemize}\itemsep 0pt
    \item $\deg_{1,1,d}({A}) = O(\sfrac{d}{\epsilon^2})$ and $|S| = \Omega(\epsilon^2 q^2)$,
    \item $\sum_{x \in S} \epsilon_x = \Omega(\epsilon^3 \cdot q^2)$,
    \item for every $(a,b) \in S$, we have $\epsilon_{(a,b)} = \Pr_{\ell \ni (a,b)}[P_\ell(a,b) = f(a,b)] \geq \sfrac{\epsilon}{2}$,
    \item for every $(a,b)\in S$, we have ${A}(a,b,f(a,b)) = 0$.
    \item $\hpartial_z({A})$ and $\disc_z({A})$ are not identically zero.
  \end{itemize}
\end{restatable}

The proof is mostly along the lines of Arora and Sudan~\cite{AroraS2003} but with a tighter analysis and some care for fields of small characteristic. We present it in its entirety, in \cref{sec:bivariate-interpolation}, for completeness and to make the changes clearer to follow. 

\medskip

We now use \cref{lem:interpolation-step} to complete the proof of \cref{thm:LDT-low-agreement-poly-agreement}.

\begin{proof}[Proof of \cref{thm:LDT-low-agreement-poly-agreement}]
Suppose $f:\F^2 \rightarrow \F$ is given and we know that $f$ passes $\LDT_d$ with probability at least $\epsilon$. By \cref{lem:interpolation-step}, we have a set $S \subseteq \F^2$ with $|S| = \gamma \cdot q^2 = \Omega(\epsilon^2 \cdot q^2)$ and a non-zero polynomial $A(x,y,z)$ with $\deg_{1,1,d}(A) = D = O(\sfrac{d}{\epsilon^2})$ such that for all $(a,b) \in S$ we have 
\begin{itemize}\itemsep0pt
  \item $A(a,b,f(a,b)) = 0$,
  \item $\Pr_{\ell \ni (a,b)}[P_\ell(a,b) = f(a,b)] \geq \sfrac{\epsilon}{2}$. 
\end{itemize}
We also know that $\hpartial_z(A)$ and $\disc_z(A)$ are both non-zero polynomials. We would like to find an $(a,b) \in S$ such that the following guarantees hold:
\begin{enumerate}[label=(\roman*)]\itemsep 0pt
\item \label{item:ldt-guarantee-pencil-nondegenerate} $\Gamma_A(a,b) \neq 0$ where $\Gamma_A(x,y) = \operatorname{disc}_z(A) = \Res_z(A, \partial_z A)$,
\item \label{item:ldt-factor-on-many-lines} There are $\Omega(\epsilon^2\cdot (q^2 - 1))$ choices for $(\alpha, \beta) \in \F^2 \setminus \set{(0,0)}$, such that there is a degree $d$ polynomial $P_{\alpha,\beta}(t)$ such that 
\[
  A(a + \alpha t, b + \beta t, P_{\alpha, \beta}(t)) = 0
\]
with $P_{\alpha, \beta}(0) = f(a,b)$, and $P_{\alpha, \beta}(a + \alpha t, b + \beta t) = f(a + \alpha t, b + \beta t)$ for at least $\Omega(\epsilon^2 \cdot q)$ many $t \in \F$. 
\end{enumerate}
If we can find such an $(a,b)$, then \cref{lem: decoding on good pencils} would imply that there is a polynomial $Q(x,y)$ of degree at most $d$ such that $A(x,y,Q(x,y)) = 0$ with $Q|_\ell = P_{\alpha, \beta}(t)$ where $\ell$ is the line $\setdef{(a + \alpha t, b + \beta t)}{t\in \F}$ for any $(\alpha, \beta)$ satisfying \cref{item:ldt-factor-on-many-lines}. Thus, for any $(a',b')  = (a + \alpha t, b + \beta t)$ that satisfy $P_{\alpha, \beta}(t) = f(a',b')$, we have $Q(a',b') = P_{\alpha, \beta}(t) = f(a', b')$. Thus, $Q$ agrees with $f$ on $\Omega(\epsilon^4 \cdot q^2)$ locations, as claimed by the theorem. \\

We now work towards guaranteeing \cref{item:ldt-guarantee-pencil-nondegenerate,item:ldt-factor-on-many-lines}. Note that we have $q^2 + q$ lines in $\F^2$, and each point $x \in \F^2$ has $q+1$ lines containing it, and each line has $q$ points on it. Consider the following bipartite graph where the set of left vertices is $\F^2$ and the set of right vertices are all lines in $\F^2$. We connect a point $x \in \F^2$ on the left to a line $\ell$ on the right if $x\in S$, $x \in \ell$ and $P_\ell(x) = f(x)$. Note that every $x \in S$ has $\epsilon_x \cdot (q+1)$ edges in this bipartite graph. If $\mu > 0$ such that $\mu q^2 = \sum_{x\in S} \epsilon_x = \Omega(\epsilon^3 \cdot q^2)$ (by \cref{lem:interpolation-step}), the number of edges in this graph is $(q+1) \cdot \sum_{x\in S} \epsilon_x = \mu \cdot (q^3+q^2)$. 

For a line $\ell$, let $y_\ell = \sfrac{\deg(\ell)}{q}$ where $\deg(\ell)$ refers to the degree in this graph. Since each $\sigma_\ell \in [0,1]$ and $\sum \sigma_\ell = \mu (q^2 + q)$. By \cref{lem:averaging}, if $L = \setdef{\ell}{\sigma_\ell > \sfrac{\mu}{2}}$ (the `high-degree lines'), then at least $\sfrac{\mu}{2} \cdot (q^3 + q^2)$ edges of the graph are incident on lines in $L$ --- let us refer to these edges as ``marked edges''. Again by \cref{lem:averaging}, there are at least $\sfrac{\mu}{4}\cdot q^2$ vertices in the left that have at least $\sfrac{\mu}{4} \cdot (q+1)$ ``marked edges'' incident on it --- let us call this set of vertices $S' \subseteq S$. 

Thus we now have a set $S' \subseteq S \subseteq \F^2$ with $|S'| \geq \sfrac{\mu}{4} \cdot q^2 = \Omega(\epsilon^3 \cdot q^2)$ such that each $x\in S'$ such that at least $\sfrac{\mu}{4}$-fraction of lines $\ell$ through it satisfy the following two properties:
\begin{itemize}\itemsep 0pt
\item $P_\ell(x) = f(x)$,
\item the line $\ell$ contains at least $\sfrac{\mu\cdot q}{2}$ points of $S$.
\end{itemize}

\paragraph{Guaranteeing \cref{item:ldt-guarantee-pencil-nondegenerate}:} The polynomial $A(x,y,z)$ has $\deg_{1,1,d} \leq D = O(\sfrac{d}{\epsilon^2})$ and hence we have $d_z = \deg_z(A) = O(\sfrac{1}{\epsilon^2})$. Therefore, the $z$-discriminant of $A$, namely $\disc_z(A)$ is the determinant of a $O(d_z) \times O(d_z)$ matrix, each of whose entries is a polynomial in $x,y$ of degree at most $D$. We recall from the last item of \cref{lem:interpolation-step} that $\disc_z(A)$ is not identically zero. 

Thus, $\deg(\disc_z(A)) = O(D \cdot d_z) = O(\sfrac{d}{\epsilon^4})$. By \cref{lem: SZ}, there are at most $O(\sfrac{dq}{\epsilon^4})$ points in $\F^2$ that the polynomial $\Gamma$ vanishes on. Since $|S'| = \Omega(\epsilon^3 \cdot q^2)$ and $q > C \cdot \sfrac{d}{\epsilon^7}$ for a large enough constant $C$, there must be some $(a,b) \in S'$ that does not make $\Gamma$ zero. 

\paragraph{Guaranteeing \cref{item:ldt-factor-on-many-lines}:} Let $(\alpha, \beta) \in \F^2 \setminus{(0,0)}$ such that $\ell = \setdef{(a + \alpha t, b + \beta t)}{t\in \F}$ is one of the `high-degree' lines in $L$ containing $(a,b)$. Let $P_\ell$ be the best-fit degree $d$ polynomial on this line. Since $(a,b) \in S$, we have $P_\ell(0) = f(a,b)$. 

Let the polynomial $B(t)$ be defined as $B(t) := A(a + \alpha t, b + \beta t, P_\ell(t))$. Clearly, $B$ has degree at most $D$. For any $t \in \F$ such that $(a',b') = (a + \alpha t, b + \beta t) \in S$, we have $P_\ell(t) = f(a', b')$ and hence $B(t) = A(a', b', f(a',b')) = 0$. Since there are at least $\sfrac{\mu \cdot q}{2} = \Omega(\epsilon^2 \cdot q)$ such points, the polynomial $B(t)$ has $\Omega(\epsilon^2 \cdot q)$ roots but has degree at most $D = O(d/\epsilon^2)$. Since $q > C \cdot \sfrac{d}{\epsilon^7}$, we must have that $B(t) = A(a + \alpha t, b + \beta t, P_\ell(t))$ is identically zero. 

As $(a,b)$ is adjacent to $\sfrac{\mu}{4} \cdot (q+1)$ lines, there are at least $\sfrac{\mu}{4} \cdot (q+1) \cdot (q-1) = \sfrac{\mu}{4} \cdot (q^2 - 1)$ choices for $(\alpha, \beta) \in \F^2 \setminus {(0,0)}$ such that there is some degree $d$ polynomial $P_{\alpha, \beta}$ such that 
\[
  A(a + \alpha t, b + \beta t, P_{\alpha, \beta}(t)) = 0.
\]
Thus, we can now invoke \cref{lem: decoding on good pencils} to get the desired conclusion. This completes the proof of \cref{thm:LDT-low-agreement-poly-agreement}.
\end{proof}

\subsection{Some consequences of Theorem \ref*{thm:LDT-low-agreement-poly-agreement}}
Given \cref{thm:LDT-low-agreement-poly-agreement}, we conclude this section with a few variants that will be useful for us in the proof of the high dimensional case. These variants, despite being seemingly stronger in appearance follows immediately from \cref{thm:LDT-low-agreement-poly-agreement} via standard techniques in the low-degree testing literature. For the sake of completeness, we include formal statements and proofs in \cref{sec:LDT-variants}.

\begin{theorem}[List-decoding version of the bivariate LDT]
  \label{thm:bivariate-list-decoding-LDT}
  For every $\epsilon_0 \in (0,1)$, finite field $\F_q$ and degree $d$ satisfying $\epsilon_0 > \Omega((d/q)^{\sfrac{1}{16}})$, the following holds:
  \begin{quote}
    For every $f:\F_q^2 \rightarrow \F_q$ and every $\epsilon > \epsilon_0$, there is a (possibly empty) set $\set{Q_1,\ldots, Q_t}$ of polynomials of degree at most $d$ such that $t\leq 2/\epsilon^8$ and $\agree(f, Q_i) \geq \epsilon^8$ for all $i \in [t]$ and 
    \[
    \Pr_{a\in \F_q^2}\insquare{\text{$a$ is $\epsilon$-good and $f(x)\notin \set{Q_1(x), \ldots, Q_t(x)}$}} \leq \epsilon_0.
    \]
  \end{quote}
\end{theorem}

\noindent The other is the `high-agreement variant', namely \cref{thm:bivariate-high-agreement-LDT}:

\bivariateldtintro*
\begin{proof}[Proofs of \cref{thm:bivariate-list-decoding-LDT,thm:bivariate-high-agreement-LDT}]
  Both the above theorems are immediate consequences of a generic reduction from the `weak form of LDT' to the `list-decoding form of LDT' (\cref{lem:weak-LDT-implies-list-decoding-LDT}) and the `high-agreement form of LDT' (\cref{lem:weak-LDT-implies-high-agreement-LDT}) used along with \cref{thm:LDT-low-agreement-poly-agreement}.
\end{proof}

\cref{thm:bivariate-high-agreement-LDT} is also of interest in the low error regime, where $\epsilon$ is close to $1$. In this case, \cref{thm:LDT-low-agreement-poly-agreement} only lets us conclude a constant fraction agreement between $f$ and a low-degree polynomial, whereas the \cref{thm:bivariate-high-agreement-LDT} asserts that $f$ agrees with a low-degree polynomial on almost all inputs. In the low error regime, in order to lift the analysis from the bivariate case to the multivariate case for our proof, we rely on the following theorem, which, as we show in the proof also follows from \cref{thm:bivariate-high-agreement-LDT}.

To avoid confusion with the low-error regime, we will use $\delta$ to denote the probability that $f$ \emph{fails} the low-degree test. 

\begin{theorem}[Low-error high-agreement bivariate LDT]\label{thm:lowerror-bi-ldt}
  There is a large enough constant $C$ such that for any finite field $\F_q$ and degree $d$ satisfying $q > C \cdot d$, and any $\delta < 0.01$ the following holds:
  \begin{quote}
    Suppose $f:\F_q^2 \to \F_q$ passes the $\LDT_d$ with probability $1 - \delta$, then there is a bivariate degree $d$ polyomial $Q$ such that 
    \[
      \Pr_{x\in \F_q^2}[f(x) = Q(x)] \geq 1 - 2\delta.
    \]
  \end{quote}
\end{theorem}
\begin{proof}
  Since $f$ passes $\LDT_d$ with probabaility at least $(1-\delta)$, we have from \cref{thm:bivariate-high-agreement-LDT} that there exists  a degree $d$ polynomial $Q$ such that \[
      \Pr_{a\in \F^2}[f(a) = {Q}(a)] \geq 1 - \delta - \epsilon_0.
  \]
  where $\epsilon_0 = \Omega\inparen{\inparen{\sfrac{d}{q}}^{\sfrac{1}{20}}}$. 
  We note that this agreement is weaker than the conclusion of the above theorem when $\delta$ is smaller than $\epsilon_0$. Let $C$ be chosen large enough so that $\delta + \epsilon_0 \leq 0.01 + \epsilon_0 < \sfrac{1}{20}\cdot (1 - \sfrac{d}{q})$.

  An important ingredient of the proof of this theorem is the definition of a \emph{corrected} version of $f$ that we denote by $f_\corr$, which is defined as follows. 
  \begin{align*}
    f_\corr(a) &:= \operatorname{plurality}_{\ell \ni a} \left\{ P^{(f,d)}_\ell(a)\right\}\, .
  \end{align*}
  where $\ell$ is a line through $a$ and $P^{(f,d)}_\ell$ denotes a univariate polynomial of degree $d$ that is closest to the restriction of $f$ on $\ell$.  

  The theorem is an immediate consequence of the following two claims

  \begin{claim}\label{clm:dist-f-fcor}
    $\Pr_{a\in \F^2}[f(a) = f_\corr(a)] \geq 1 - 2\delta$.
  \end{claim}

  \begin{claim}\label{clm:fcorr-is-low-degree} 
    For all $a \in \F^2$, we have $f_\corr (a) = Q(a)$. 
  \end{claim}
  \begin{proof}[Proof of \cref{clm:dist-f-fcor}]
    Let $B$ denote the set of \emph{bad} points in $\F^2$ defined as follows 
    \[
    B := \setdef{a \in \F^2}{\Pr_{\ell \ni a}[P^{(f,d)}_\ell(a) \neq f(a)] \geq \nicefrac{1}{2}}.
    \]
    Clearly, if $a \notin B$, then we have $f_\corr(a) = f(a)$. Hence, we have $\Pr_a[f(a) \neq f_\corr(a)] \leq \Pr_a[a \in B]$. On the other hand, 
    \begin{align*}
      \delta & \geq \Pr_{x,\ell}[\text{$f$ fails $\LDT_d$ on $(x,\ell)$}]\\ 
      & \geq \Pr[a \in B] \cdot \Pr_{\substack{a\in \F^2\\\ell \ni a}}[P^{(f,d)}_\ell(a) \neq f(a) \mid a \in B]\\
      & \geq \Pr_{a \in \F^2}[f(a) \neq f_\corr(a)] \cdot \sfrac{1}{2}
    \end{align*}
    which implies that $\Pr_a[f(a) \neq f_\corr(a)]\leq 2\delta$ and hence $\Pr_a[f(a) = f_\corr(a)] \geq 1 - 2\delta$.
  \end{proof}

  \begin{proof}[Proof of \cref{clm:fcorr-is-low-degree}]
    As mentioned earlier, we know that $f$ has fractional agreement of at least $(1-\delta - \epsilon_0)$ with the degree $d$ polynomial $Q$. Let $B$ be the set of inputs where $f$ and $Q$ disagree. Clearly, $|B| < q^2(\delta + \epsilon_0)$. Let $a \in \F^2$ be an arbitrary input. Since lines through $a$ are a uniform cover of the space $\F^2 \setminus \{a\}$, we have that for all large enough $q$,
    \[
    \E_{\ell \ni a}[|\ell \cap B|] \leq q\cdot \left(\frac{q^2(\delta + \epsilon_0)}{q^2 -1} \right) \leq 2q(\delta + \epsilon_0).
    \]
    Thus, by Markov's inequality, $\Pr_{\ell \ni a}[|\ell \cap B| > 8q(\eta + \epsilon)] < \sfrac{1}{4}$.

    Now, let $\ell$ be a line through $a$  such that $|\ell \cap B| < 8q(\delta + \epsilon_0)$. For every such line $\ell$, the disagreement between $Q$ and $f$ on $\ell$ is at most $8q(\delta + \epsilon_0) + 1\leq 10 q(\delta + \epsilon_0) < \sfrac{1}{2} \cdot (1 - \sfrac{d}{q})$. Since this is less than half the minimum distance of a Reed-Solomon code of degree $d$ and block-length $q$, we have that $P_\ell^{(f,d)}$ must equal $Q\mid_\ell$. 

    Hence, we have that on $\sfrac{3}{4}$-th of the lines $\ell$ through $a$, the restriction $Q\mid_{\ell}$ equals the polynomial $P_{\ell}^{(f,d)}$. From the definition of $f_\corr$, it now follows that $f_\corr(a) = Q(a)$. Since $a$ is an arbitrary point in $\F^2$, we have that the $f_\corr$ and $Q$ agree everywhere on $\F^2$.
  \end{proof}
  \noindent
  That concludes the proof of \cref{thm:lowerror-bi-ldt}.
\end{proof}

\section{Lifting to $m$-variate low-degree tests}\label{sec:bootstrapping}

In this section, we bootstrap the low-degree test from 2 dimensions to $m$ dimensions.  This bootstrapping is inspired and adapted from the corresponding bootstrapping arguments due to Friedl-Sudan \cite{FriedlS1995}, Arora-Sudan \cite{AroraS2003} and Bhangale-Dinur-Navon \cite{BhangaleDN2017} and the local-list-decoder for multivariate polynomial codes due to Sudan, Trevisan \& Vadhan \cite{SudanTV2001}. The bootstrapping arguments heavily use the expansion properties of the points-lines-planes affine Grassmannian. So, we first mention some preliminaries that we would need..

\subsection{Preliminaries}

\begin{theorem}[(Weak) Johnson Bound]\label{thm:johnson}For any function $f\colon \F \to \F$, degree parameter $d$ and $\eps \in (0,1)$, let $P_1,\dots, P_r$ be the set of all degree-$d$ polynomials that have agreement at least $\eps$ with $f$ (i.e, $\Pr_{x \in \F}[f(x)=P(x)]\geq \eps$). 
  We say that a point $x\in \F$ is non-unique with respect to function $f$, degree $d$ and agreement $\eps$, if there exist two \emph{distinct} polynomials $P_k \neq P_{k'}$, $1 \leq k, k' \leq r$ such that $P_k(x) = P_{k'}(x)$ or in short, ``$x \in \nonunique^{d}_\eps(f)$''.

  If $\eps \geq 2\sqrt{\nicefrac{d}{q}}$, then the following bounds hold.
\begin{enumerate}
  \item $r \leq \nicefrac{2}{\eps}$.\label{item:johnson}
  \item  The number of points in $\nonunique^{d}_\eps(f)$ is at most $\binom{r}{2}\cdot d \leq \nicefrac{2d}{\eps^2}$.\label{item:nonunique}
\end{enumerate}
\end{theorem}

\begin{definition} Let G = $(A, B, E)$ be a bi-regular bipartite graph, and let $M \in  \Real^{A \times B}$ be the adjacency matrix normalized such that $\|M 1\| = 1$, denote by $\lambda(G)$ the value
\[ \lambda(G) = \max_{\vv\perp \mathbbm{1}}\left\{ \frac{\|M\vv\|}{\|\vv\|}\right\}.\]
We will refer to $\lambda(G)$ as the second eigen-value of $G$.
\end{definition}

This is really the second largest singular value of M , with a different normalization (such that the maximal singular value equals 1).

The following is the classic Expander Mixing Lemma (for bipartite graphs).
\begin{lemma}[Expander Mixing Lemma]\label{lem:bipartite-eml} Let $G=(A,B,E)$ be a biregular bipartite graph with second eigen-value $\lambda$. Then for any two functions $g \colon A \to \R$ and $h\colon B \to \R$ with means and variances $\mu_g, \mu_h$ and $\sigma_g^2$ and $\sigma_h^2$, we have 
    \[ \left|\Pr_{(a,b)\in E} [g(a)\cdot h(b)] - \mu_g \cdot \mu_h \right| \leq \lambda \cdot \sigma_g \cdot \sigma_h \]
 \end{lemma}

 We will be using the following corollary of the above Expander Mixing Lemma. 

\begin{lemma}[Corollary of Expander Mixing Lemma {\cite[Lemma 6]{BhangaleDN2017}}]\label{lem:bipartite-eml-withbadset} Let $G=(A,B,E)$ be a biregular bipartite graph with second eigen-value $\lambda$. Then for any set $A' \subseteq A$ of measure $\mu$ and any $E' \subseteq E$, we have
    \[ \left|\Pr_{a \sim A', b \sim N_G(a)}[(a,b) \in E'] - \Pr_{b \in B, a\sim N_G(b) \cap A'}[(a,b) \in E']\right| \leq \nicefrac{\lambda}{\sqrt{\mu}}. \] 
 \end{lemma}

The following are well-known properties of the points-lines-planes affine Grassmannian in $\F_q^m$. 

\begin{theorem}\label{thm:spectralgraphs}
  The second eigen-value of the following biparite inclusion graphs is as follows:
  \begin{enumerate}
      \item \label{spec:point-line-in-plane-graph} For all $m \geq 2$ and $G = G(\F^m, \calL^{(m)})$, $\lambda(G) = \nicefrac{1}{\sqrt{q}}$. 
      \item \label{spec:point-plane-graph} For all $m\geq 2$, $G = G(\F^m,\calP^{(m)})$, $\lambda(G) = \nicefrac{1}{q}$.
      \item \label{spec:line-plane-graph} For all $m \geq 2$, $G = G(\calL^{(m)},\calP^{(m)})$, $\lambda(G) = \frac{1}{\sqrt{q}} \cdot (1 + o(\nicefrac{1}{\sqrt{q}}))$.
      \item \label{spec:lines-planes-graph-containing-x} For all $m\geq 2$ and $G=G(\calL^{(m)}_x,\calP^{(m)}_x)$, $\lambda(G) = \nicefrac{1}{\sqrt{q+1}}$.
  \end{enumerate}
\end{theorem}

\subsection{The bootstrapping statements}

Given a function $f \colon \F^m \to \F$ and a line $\ell$ in $\F^m$, we let $P^{(f,d)}_\ell$ be the best fit degree-$d$ univariate polynomial (presented as a list of evaluations) that agrees with $f$ on the line $\ell$ (if there is more than one such polynomial, we break ties arbitrarily). We now, define the following quantities for any function $f\colon \F^m \to \F$, line $\ell$ and plane $\pi$:
\begin{align}
    \delta_f(\ell) &:= \Pr_{x \in \ell}[P^{(f,d)}_\ell(x) \neq f(x)], \nonumber\\
    \delta_f(\pi)  & := \E_{\ell \in \pi}[\delta_f(\ell)] = \Pr_{\substack{\ell \in \pi\\x \in \ell}}[P^{(f,d)}_\ell(x) \neq f(x)],\nonumber\\
    \delta_f &:= \E_\pi[\delta_f(\pi)] = \Pr_{\substack{\ell \in \F^m\\x \in \ell}}[P^{(f,d)}_\ell(x) \neq f(x)] \label{eq:deltaf}.
\end{align} 

\subsubsection*{Bootstrapping in the low-error regime}

In \cref{sec:lowerrorbootstrapping}, we  bootstrap the low-error bivarariate LDT \cref{thm:lowerror-bi-ldt} to prove \cref{thm:main-lowerror-ldt}, a similar theorem for $m$ dimensions. This bootstrapping argument is an adaptation of a similar argument due to Friedl and Sudan \cite{FriedlS1995}, the main difference being that we use the expansion properties of the points-lines-planes affine Grassmannian rather than that of a more complicated points-lines-surfaces HDX used in \cite{FriedlS1995}.

\multivariatelowerrorLDT*

\subsubsection*{Bootstrapping in the high-error regime}

In the subsequent \cref{sec:higherrorbootstrapping}, we perform a similar bootstrapping for the high-error regime to obtain the following theorem from bivariate low-degree test (\cref{thm:LDT-low-agreement-poly-agreement}). This bootstrapping argument is inspired from the corresponding arguments due to Arora-Sudan \cite{AroraS2003} and Bhangale-Dinur-Navon \cite{BhangaleDN2017}. The argument presented here, while elementary, requires a careful analysis using the spectral properties (repeated applications of \cref{lem:bipartite-eml,lem:bipartite-eml-withbadset}) of various subgraphs of the points-lines-planes affine Grassmannian. The key improvement from \cite{AroraS2003} and \cite{BhangaleDN2017} is that we bootstrap from the base case of $m=2$ while previous arguments worked with a base case of at least $m \geq 3$. 

\begin{theorem}[high-error regime $m$-variate LDT]\label{thm:higherror-ldt}
  For every $\epsilon_0 \in (0,1)$, and finite field $\F_q$ and degree $d$ satisfying $\epsilon_0 > \Omega((d/q)^{\sfrac{1}{48}})$, the following holds. 
    \begin{quote}
    If the points table $f \colon \F^m \to \F$ and degree-$d$ lines oracle $\ell \mapsto P_\ell$ satisfy 
    \[\Pr_{x, \ell \ni x}[f(x) = P_\ell(x)] \geq  5\eps, \]
    then there exists an $m$-variate degree-$d$ polynomial $Q$ such that $\Pr_x[f(x)=Q(x)] \geq \eps^2.$
    \end{quote}  
\end{theorem}

As in the bivariate case, the above `weak form of the low-degree test' (\cref{thm:higherror-ldt}) can be reduced to the `high-agreement form of the low-degree test' (\cref{thm:multivariate-low-degree-test-intro}) using standard transformations (\cref{lem:weak-LDT-implies-high-agreement-LDT}). 

\subsection{Low-error regime: Proof of {\cref*{thm:main-lowerror-ldt}}}\label{sec:lowerrorbootstrapping}

Given a function $f\colon \F^m \to \F$, we define the self-correction $f_\corr$ of the function $f$ as follows:
\begin{align*}
    f_\corr(x) &:= \operatorname{plurality}_{\ell \ni x} \left\{ P^{(f,d)}_\ell(x)\right\}.
\end{align*}

Let $\eps_0,\delta_0 \in (0,1)$, the field $\F_q$ and degree $d$ satisfy $\eps_0 \geq \Omega(\left(\nicefrac{d}{q}\right)^{\nicefrac{1}{20}})$ as in the hypothesis of . Define $\delta_0 := \nicefrac{1}{20}(1-\nicefrac{d}{q})-\eps_0$. By the hypothesis of \cref{thm:main-lowerror-ldt}, we have that $\delta_f \leq \nicefrac{\delta_0}{2}$.

The theorem follows from the two claims.
\begin{claim}\label{clm:fandfcorrareclose} 
    $\delta(f,f_\corr) \leq 2\delta_f$.
\end{claim}

\begin{claim}\label{clm:deltafhalves}
    If $q > \nicefrac{100}{\delta_0^2}$ and $\delta_f \leq \nicefrac{\delta_0}{2}$, then $\delta_{f_\corr} \leq \nicefrac{\delta_f}{2}$.
\end{claim}
\begin{proof}[Proof of \cref{thm:main-lowerror-ldt}]
Define a sequence of functions $f^{(0)}, f^{(1)}, \ldots \colon \F^m \to \F$ as follows: $f^{(0)} := f$ and for $i \geq 1$, $f^{(i)} := \left(f^{(i-1)}\right)_\corr$. We then have, from \cref{clm:deltafhalves}, that $\delta_{i} := \delta_{f^{(i)}} \leq \nicefrac{\delta_f}{2^i}$. For any fixed $q$ and $m$, since $\delta_f$ can only take a set of finitely many values, we have that $\delta_i$ eventually becomes 0, i.e., there exists $i_*$ such that $\delta_{i_*} = \delta_{f^{(i_*)}} = 0$. Then, $f^{(i_*)}$ is a degree-$d$ $m$-variate polynomial, say $Q$. Now, by \cref{clm:fandfcorrareclose}, we have
\begin{align*}
    \delta(f,Q) 
    = \delta\left(f^{(0)},f^{(i_*)}\right) 
    \leq \sum_{i=1}^{i_*}\delta\left(f^{(i-1)},f^{(i)}\right)
    \leq \sum_{i=1}^{i_*} 2 \delta_{f^{(i-1)}} 
    \leq  2 \sum_{i=1}^{i_*} \nicefrac{\delta_f}{2^{i-1}}\leq 4\delta_f.
\end{align*}
This completes the proof of \cref{thm:main-lowerror-ldt} (assuming the two claims).
\end{proof}

We now proceed to prove the two claims. 

\begin{proof}[Proof of \cref{clm:fandfcorrareclose}]
Given the points table $f \colon \F^m \to \F$ and lines oracle $\ell \mapsto P^{(f,d)}_\ell$, we define the set of ``bad'' points as follows: $\BAD := \{ x \in \F^m \colon \Pr_{\ell \ni x}[P^{(f,d)}_\ell(x) \neq f(x)] \geq \nicefrac{1}{2} \}$.  
Clearly, if $x \notin \BAD$, we have $f_\corr(x) = f(x)$. Hence, $\delta(f,f_\corr) \leq \Pr[ x \in \BAD]$. On the other hand, we have
\begin{align*}
    \delta_f &= \Pr_{\substack{x\in \F^m\\\ell \ni x}}[P^{(f,d)}_\ell(x) \neq f(x)] \\
    &\geq \Pr[x \in \BAD] \cdot \Pr_{\substack{x\in \F^m\\\ell \ni x}}[P^{(f,d)}_\ell(x) \neq f(x) \mid x \in \BAD]\\
    &\geq \delta(f,f_\corr) \cdot \nicefrac{1}{2}.  
\end{align*}
Hence, $\delta(f,f_\corr) \leq 2\delta_f$.
\end{proof}

\subsubsection{Self-correction passes LDT with better probability (Proof of \cref*{clm:deltafhalves})}

This is the heart of the bootstrapping argument, where we use the (high-dimensional) expansion of the points-lines-planes Affine Grassmannian complex to show that the self-corrected function $f_\corr$ passes the LDT with even better probability than the original function $f$. 

We begin by showing the following bound on $\delta_{f_\corr}$.
\begin{align*}
    \delta_{f_\corr} 
    &= \Pr_{x, \ell \ni x}[ P^{(f_\corr,d)}_\ell(x) \neq f_\corr(x)]\\
    & \leq \Pr_{x,\ell \ni x}[ P^{(f,d)}_\ell(x) \neq f_\corr(x)]\\
    & \leq \Pr_{\substack{x\\\ell,\ell' \ni x}}[P^{(f,d)}_\ell(x) \neq P^{(f,d)}_{\ell'}(x)].
\end{align*}
The first inequality follows since $P^{(f,d)}_\ell$ cannot perform any better than the best-fit degree-$d$ polynomial $P^{(f_\corr,d)}_\ell$. The second inequality follows since for each $x\in \F^m$, $f_\corr(x)$ is the most popular value among $P^{(f,d)}_\ell(x)$ as $\ell$ varies over all lines $\ell$ through $x$ and hence the probability (over $\ell)$ that $f_\corr(x)= P^{(f,d)}_\ell(x)$ is at least the collision probability that for two independently chosen lines $\ell, \ell'$ through $x$, we have $P^{(f,d)}_\ell(x) = P^{(f,d)}_{\ell'}(x)$. 

It thus suffices to bound the probability that $P^{(f,d)}_\ell(x) = P^{(f,d)}_{\ell'}(x)$ where $x, \ell, \ell'$ are chosen as follows: $x$ is picked uniformly from $\F^m$, $\ell, \ell'$ are independently chosen to be two lines through $x$ in $\F^m$. An equivalent way of picking this triple is first picking a random plane $\pi$ in $\F^m$, a point $x$ in the plane $\pi$ and two independent lines $\ell, \ell'$ in the plane $\pi$ that contain $x$. 
We now define three (bad) events E1, E2 and E3 (based on the choice of $\pi, \ell, \ell')$ such that (1) if none of the three events happen, then $P^{(f,d)}_\ell(x) = P^{(f,d)}_{\ell'}(x)$ and (2) the probability of each event is at most $\delta_f/6$. This will complete the proof of the claim.

\begin{description}
    \item[Event E1$(\pi)$:] $\delta_f(\pi) \geq \delta_0$.
    
        Consider the bipartite lines-planes incidence graph $G(\calL^{(m)},\calP^{m})$ in $\F^m$ which has second eigen-value at most $\nicefrac{2}{\sqrt{q}}$ (by \cref{thm:spectralgraphs}-\ref{spec:line-plane-graph}). We know that $\delta_f(\pi) = \E_{\ell \in\pi}[\delta_f(\ell)]$ and $\delta_f = \E_\ell[\delta_f(\ell)]$.  Consider the functions $g \colon \calL^{(m)} \to \R$ and $h\colon \calP^{(m)} \to \R$ defined as follows: $g(\ell) := \delta_f(\ell)$ and $h(\pi) := \mathbbm{1}[\delta_f(\pi)\geq \delta_0]$. 
        These functions satisfy $\mu_g = \delta_f$, $\mu_h = \Pr[\text{E1}] =: \mu$, $\sigma^2_g = \Var_\ell[\delta_f(\ell)] \leq \delta_f$ and $\sigma^2_h = \Var_\pi[\mathbbm{1}[\delta_f(\pi)\geq \delta_0]] \leq \mu$. Applying \cref{lem:bipartite-eml} to the graph $G$ with functions $g$ and $h$ as defined above, we have
        \[ \mu \cdot \delta_0 - \mu\cdot \delta_f \leq \lambda \cdot \sqrt{\mu \cdot \delta_f}. \]
        Equivalently, $\mu \leq(\nicefrac{\lambda}{(\delta_0-\delta_f)})^2 \cdot \delta_f$. Choosing $q \geq \nicefrac{100}{\delta_0^2}$, we have $\mu \leq \nicefrac{\delta_f}{6}$ (since $\lambda \approx \nicefrac{1}{\sqrt{q}}$ and $\delta_f \leq \delta_0/2$).

        Furthermore, if event E1 does not happen then by \cref{thm:lowerror-bi-ldt} (the bivariate LDT in the low-error regime), we have that there exists a bivariate degree-$d$ polynomial $Q_\pi$ on the plane $\pi$ such that $\delta(f|_\pi, Q_\pi) \leq 2\delta_f(\pi) \leq 2\delta_0$. 

    \item[Event E2$(\pi,\ell)$:] $\neg$E1 and $\delta(f|_\ell, Q_\pi|_\ell) \geq 4\delta_0$. (Here $Q_\pi$ is the bivariate degree-$d$ polynomial that is guaranteed to exist since E1 does not occur.)

        To begin with let us fix a plane $\pi$ such that E1 does not occur. We will later randomize over the choice of the  plane. Since  E1 does not occur, we have there exists a bivariate degree-$d$ polynomial such that $\delta(f|_\pi, Q_\pi) \leq 2\delta_f(\pi)\leq 2\delta_0$. 

        Consider the bipartite points-line incidence graph $G(\pi,\calL^{(\pi)})$ in the plane $\pi$  which has second eigen-value at most $\nicefrac{1}{\sqrt{q}}$ (by \cref{thm:spectralgraphs}-\ref{spec:point-line-in-plane-graph}). Let $\BAD_\pi$ be the set of lines $\ell$ in $\pi$ such that $\delta(f|_\ell,Q_\pi|_\ell) \geq 4\delta_0$. Consider the functions $g \colon \pi \to \R$ and $h\colon \calL^{(\pi)} \to \R$ defined as follows: $g(x) := \mathbbm{1}[f(x)=Q_\pi(x)]$ and $h(\ell) := \mathbbm{1}[\ell \in \BAD_\pi]$. 
        These functions satisfy $\mu_g = \delta(f|_\pi,Q_\pi)$, $\mu_h = \mu(\BAD_\pi) =: \mu_\pi$, $\sigma^2_g = \Var_x[\mathbbm[f(x)\neq Q_\pi(x)]]\leq \mu_g =\delta(f|_\pi,Q_\pi)$ and $\sigma^2_h = \Var_\ell[\mathbbm{1}[\ell \in \BAD_\pi]] \leq \mu_\pi$. Applying \cref{lem:bipartite-eml} to the graph $G$ with functions $g$ and $h$ as defined above, we have
        \[ \mu_\pi \cdot 4\delta_0 - \mu\cdot \delta(f|_\pi,Q_\pi) \leq \lambda \cdot \sqrt{\mu_\pi \cdot \delta(f|_\pi,Q_\pi)}. \]
        Equivalently, $\mu_\pi \leq(\nicefrac{\lambda}{(4\delta_0-\delta(f|_\pi,Q_\pi))})^2 \cdot \delta(f|_\pi,Q_\pi) \leq (\nicefrac{\lambda^2}{2\delta_0^2})\cdot \delta_f(\pi)$ (since $\delta(f|_\pi,Q_\pi) \leq 2\delta_f(\pi) \leq 2\delta_0$). Choosing $q \geq \nicefrac{3}{\delta_0^2}$, we have that for this choice of $\pi$, $\mu_\pi \leq \nicefrac{\delta_f(\pi)}{6}$ (since $\lambda \approx \nicefrac{1}{\sqrt{q}}$ and $\delta_f \leq \delta_0/2$).

        We now average over $\pi$ as follows:
        \begin{align*}
            \Pr_{\pi,\ell}[\text{E2}] = \E_\pi\left[\mathbbm{\neg\text{E1}(\pi)} \cdot \mu_\pi \right] \leq \E_\pi\left[\mathbbm{\neg\text{E1}(\pi)} \cdot (\nicefrac{\delta_f(\pi)}{6}) \right] \leq \E_\pi[\nicefrac{\delta_f(\pi)}{6}] = \nicefrac{\delta_f}{6}.
        \end{align*}

        Observe that if for a particular choice of random $\pi$ and $\ell$, events E1 and E2 do not occur, then there exists a bivariate degree-$d$ polynomial $Q_\pi$ such that $\delta(f|_\pi, Q_\pi) \leq 2\delta_f(\pi) \leq 2\delta_0$ and furthermore $\delta(f|_\ell, Q_\pi|_\ell) \leq 4\delta_0$. If $1-8\delta_0 \geq \nicefrac{d}{q}$, then $Q_\pi|_\ell$ is the (unique) best-fit degree-$d$ polynomial to $f|_\ell$, i.e., $P^{(f,d)}_\ell = Q_\pi|_\ell$.

    \item[Event E3$(\pi,\ell')$:] $\neg$E1 and $\delta(f|_{\ell'}, Q_\pi|_{\ell'}) \geq 4\delta_0$.
    
        This event is identical to E2 and hence $\Pr_{\pi,\ell'}[\text{E3}] \leq \delta_f/6$. 
\end{description}  

Clearly, if events E1, E2 and E3 do not occur (for a particular choice of $\pi, \ell, \ell'$ and $x$), we have that there exists a degree-$d$ bivariate polynomial $Q_\pi$ such that $P^{(f,d)}_\ell = Q_\pi|_\ell$ and  $P^{(f,d)}_{\ell'} = Q_\pi|_{\ell'}$. Hence, $P^{(f,d)}_\ell(x) = Q(x)= P^{(f,d)}_{\ell'}(x)$. This completes the proof of the claim. \qed

\subsection{High-error regime: Proof of \cref*{thm:higherror-ldt}}\label{sec:higherrorbootstrapping}

A degree-$d$ lines oracle assigns to each line in $\F^m$, a a degree-$d$ polynomials (presented as a list of evaluations, i.e., a Reed-Solomom codeword) or $\bot$. 
We say that the degree-$d$ lines oracle is \emph{$\eps$-well-behaved with respect to the function $f\colon \F^m \to \F$} if for every line $\ell$, $P_\ell$ is well-defined (i.e, $\neq \bot$) if and only if $\Pr_{x\in \ell}[f(x) = P_\ell(x)] \geq \eps$. 

We say a point $x \in \F^m$ is \emph{$\eps$-good} if it agrees with at least an $\eps$-fraction of lines that pass through it, i.e.,
\[ \Pr_{\ell \ni x}[f(x) = P_\ell(x)] \geq \eps. \] We will refer to the set of $\eps$-good points in $\F^m$ as $\eps\text{-GOOD}$. We will need the above notation, both when the ambient dimension is $m=2$ (i.e., a plane) or general $m$ (i.e., $\F^m$). To distinguish these two cases, in the former we say ``$x$ is $\eps$-good with respect to plane $\pi$'' , while in the latter we just say ``$x$ is $\eps$-good''. Sometimes, we also say ``$x$ is $\eps$-locally-good'' vs. ``$x$ is $\eps$-globally-good''. 

We say that ``$x$ is $\eps$-explained with respect to a plane $\pi$'' if there exists an bivariate degree-$d$ polynomial $Q$ (defined on the plane $\pi$) such that (1) $\Pr_{z\in \pi}[f(z)=Q(z)] \geq \eps$ and (2) $f(x)=Q(x)$. 

We will assume the the bivariate LDT, given by \cref{thm:bivariate-list-decoding-LDT}, which states the following (rewrriten in the language of ``''$\eps$-explained''-ness). Let $\eps_0 := \Omega((\nicefrac{d}{q})^{\nicefrac{1}{16}})$. For every $\eps \geq \eps_0$, the following holds. For any function $f \colon \F^2 \to \F$, 
\begin{equation} \label{eq:strong-bi-ldt}
  \Pr_x[ x \text{ is $\eps$-good but not $\eps^8$-explained }] \leq \eps_0.
\end{equation}

We prove \cref{thm:higherror-ldt} by defining a self-corrected function as in the low-error regime and showing that the corrected function passes the LDT with significantly higher probability. However, unlike the low-error regime, there are several candidate corrected functions and we disambiguate among them using an additional advice $(x,\sigma) \in \F^m \times \F$. 

For $x \in \F^m$ and $\sigma \in \F$, define $f_{\corr_\delta}^{x,\sigma}\colon \F^m \to (\F\cup \bot)$ as follows: For any $y \in \F^m \setminus \{x\}$, let $P$ be the unique degree-$d$ univariate polynomial on the unique line $\ell = \ell_{x,y}$ through $x$ and $y$ such that (1) $\Pr_{z\in \ell}[f(z)=P(z)] \geq \delta$ and (2) $P(x) = f(x)$. If there is no such polynomial $P$ or there is more than one such polynomial, we set $P := \bot$. Finally, we set $f_{\corr_\delta}^{x,\sigma}(y) := P(y)$. Also, set $f_{\corr_\delta}^{x,\sigma}(x) := \sigma$.  Finally we define $f_{\corr_\delta}^{x} := f_{\corr_\delta}^{x,f(x)}$. We will refer to the function $f_{\corr_\delta}^{x}$ as the $x$-corrected function. 

The following lemma states for a random $\eps$-good $x$, the corrected function $f_{\corr_\delta}^x$ passes the LDT with very high probability $1-\gamma$. 

\begin{lemma}\label{clm:randomfcorrxpassesldt} 
  Let the field $\F$, degree $d$ and $\mu, \gamma, \eps \in (0,1)$ satisfy $q \geq 800d\cdot \max\left\{\nicefrac{1}{(\mu\gamma)^{16}}, \nicefrac{1}{\mu\gamma^2\eps^{16}}\right\}$. 
    There is a integer $C$ and $\tau \in (0,1)$ such Given a function $f$ and $\eps$-well-behaved degree-$d$ lines oracle, let $S$ be a subset of the $2\eps$-good points in $\F^m$ of density $\mu$. Then, the distribution $(x,\ell)$ obtained by picking a random point in $S$ and a random line $\ell$ in $\F^m$ satisfies
    \[ \Pr_{x \in S, \ell}\left[ \exists \text{ degree-$d$ polynomial $P_\ell$ such that } \delta(f^x_{\corr_{{\eps^8}/{2}}}|_\ell, P) \leq \gamma \right] \geq 1- 2\gamma. \]  
  \end{lemma}

We now complete the proof of \cref{thm:higherror-ldt} assuming this lemma. 

\begin{proof}[Proof of \cref{thm:higherror-ldt}]
    We are given a points table $f$ and a degree-$d$ lines oracle $\ell \mapsto P_\ell$ for $m$ dimensions such that $\Pr_{x,\ell \ni x}[f(x)=P_\ell(x)] \geq 5\eps$. We first  modify the lines oracle by setting $P_\ell$ to $\bot$ if $\Pr_{x\in \ell}[f(x) = P_\ell(x)] < \eps$. This ensures, that the lines oracle is $\eps$-well-behaved wrt $f$. This reduces the acceptance probability of the LDT by at most $\eps$ and the modified lines oracle satisfies
    \begin{equation}\label{eq:LDTassumption}
        \Pr_{x, \ell \ni x}[f(x) = P_\ell(x)] \geq 4\eps.
    \end{equation}
    We will be setting $\mu := \eps$ and $\gamma := \sfrac{\eps^2}{12}$. There exists a suitable large constant $C$,  such that $q \geq \nicefrac{C\cdot d}{\eps^{48}}$ implies that the field-size is large enough for \cref{lem:randomfcorrxpassesldt} for this choice of $\eps, \mu, \gamma$.

    For any line $\ell$, let $P^{(\ell)}_1,\dots,P^{(\ell)}_r$ be all the univariate degree-$d$ polynomials $P$ that satisfy $\Pr_{z\in \ell}[P(z)=f(z)] \geq \nicefrac{\eps^8}{2}$. Let $\nonunique^d_{\nicefrac{\eps^8}{2}}(f|_\ell)$ be the set of all points $x$ on the line $\ell$ such that there exist two distinct polynomials $P^{(\ell)}_k \neq P^{(\ell)}_{k'}$ such that $P^{(\ell)}_k(x)= P^{(\ell)}_{k'}(x)$. By \cref{thm:johnson}-\ref{item:nonunique}, we have that the number of points in $\nonunique^d_{\nicefrac{\eps^8}{2}}(f|_\ell)$ is at most $\nicefrac{8d}{\eps^{16}}$. 

    The LDT hypothesis \eqref{eq:LDTassumption} implies that
    \[ \Pr_{\ell, x\in \ell} [f(x) = P_\ell(x) \text{ and } x \notin \nonunique^{d}_{\nicefrac{\eps^8}{2}}(f_\ell)] \geq 4\eps - \nicefrac{8d}{q\cdot (\eps^8)^2} \geq 3\eps,\]
    provided $q \geq \nicefrac{8d}{\eps\cdot \eps^{16}}$. Or equivalently,
    \[ \E_x\left[\Pr_{\ell \ni x} \left[f(x) = P_\ell(x) \text{ and } x \notin \nonunique^{d}_{\nicefrac{\eps^8}{2}}(f|_\ell)\right] \right] \geq 3\eps.\]
    Define $S \subseteq \F^m$ to be the set of points as follows:
    \begin{equation*}
      S := \left\{ x \in \F^m \colon \Pr_{\ell \ni x} \left[f(x) = P_\ell(x) \text{ and } x \notin \nonunique^{d}_{\nicefrac{\eps^8}{2}}(f|_\ell)\right] \geq 2\eps \right\} \,.
    \end{equation*}
    We thus have $\Pr_{x \in \F^m}[x \in S]\geq \eps$. Also observe that any $x\in S$ is $2\eps$-good. We now apply \cref{clm:randomfcorrxpassesldt} on the set $S$ with $\mu :=\eps$ to obtain that $f^x_{\corr_{\nicefrac{\eps^8}{2}}}$ passes the LDT with probability at least $3\gamma$ for a random $x$. Fix any such $x\in S$. We have $\delta_{f_{\corr_{\nicefrac{\eps^8}{2}}}^x} \leq 3\gamma$ where $\delta_f$ is as defined in \eqref{eq:deltaf}. Applying \cref{thm:main-lowerror-ldt} ($m$-variate LDT in the low-error regime), we obtain that there exists a $m$-variate degree-$d$ polynomial $Q^x$ such that $\delta(Q^x,f_{\corr_{{\eps^8}/{2}}}^x) \leq 12\gamma$. 

    We will now argue that $f$ and $f^x_{\corr_{\nicefrac{\eps^8}{2}}}$ agree on at least $2\eps^2$ fraction of points since $x \in S$. Recall the definition of $S$ and the definition of the $x$-corrected function $f^x_{\corr_{\nicefrac{\eps^8}{2}}}$. For every line $\ell$ through $x$ such that $f(x)=P_\ell(x)$ and $x \notin \nonunique^d_{\nicefrac{\eps^8}{2}}$, we have that $f^x_{\corr_{\nicefrac{\eps^8}{2}}}(z)=P_\ell(z)$ for every $z\in \ell$ such that $P_\ell(z)= f(z)$. This is because (1) $P_\ell$ agrees with $f$ on at least $\eps\geq \nicefrac{\eps^8}{2}$-fraction of the points (since the lines oracle is $\eps$-well-behaved), (2) $P_\ell(x)=f(x)$ and (3) $P_\ell$ is the unique polynomial satisfying (1) and (2). Hence, for any such line $\ell$, we have $\Pr_{z\in \ell}[f^x_{\corr_{\nicefrac{\eps^8}{2}}}(z)=f(z)]\geq \eps$. Furthermore, there are at least $2\eps$-fraction of such lines through $x$ since $x\in S$. Hence, $\Pr_{z}[f^x_{\corr_{\nicefrac{\eps^8}{2}}}(z)=f(z)]\geq 2\eps\cdot \eps =2\eps^2$.

    Combining this, with $\delta(Q^x,f_{\corr_{{\eps^8}/{2}}}^x) \leq 12\gamma$, we have that $\Pr_{z}[Q^x(z)=f(z)]\geq 2\eps^2 -12\gamma \geq \eps^2$ provided $\gamma \leq \nicefrac{\eps^2}{12}$. This completes the proof of the theorem.
\end{proof}

\subsubsection{Corrected function passes LDT with high probability (proof of \cref*{clm:randomfcorrxpassesldt})}

Our plan is to show that if the subset $S$ of $2\epsilon$-globally good points is of density at least $\mu$, then a random point $x$ in $S$ satisfies that the corresponding  $x$-corrected function passes the low-degree test with probability $1 - O(\gamma)$ (for arbitrarily small $\gamma > 0$ provided $d/q$ is small enough). To do so we consider a random point $x\in S$ and a random line $\ell$ and consider the plane $\pi$ containing $x$ and $\ell$ and prove that with probability $1 - \gamma$ the following hold: There is a bivariate polynomial $Q_{\pi}^x$ on the plane $\pi$ such that (1) $Q_\pi^x$ has agreement at least $\epsilon^8$ with $f|_\pi$, (2) $Q_\pi^x(x) = f(x)$, and (3) For $1 - \gamma$ fraction of the points $y \in \pi$, if we let $\ell'$ be the line through $x$ and $y$, then $Q_\pi^x|_{\ell'}$ is the unique polynomial with $\epsilon^8/2$ agreement with $f|_{\ell'}$.

Note that the existence of even one point $x$ with properties (1)-(3) above suffice, but the proof essentially forces us to prove that most points in $S$ satisfy properties (1)-(3).

We start with a preliminary statement. Note that the distribution on the triples $(x,\ell,\pi)$ of interest to us is the following: $D_1$ is the distribution obtained by picking $x$ to be a uniform point in $S$, $\ell$ a uniformly random line in $\F^m$ and $\pi$ be a uniform plane containing $\ell$ and $x$ (w.h.p. $\pi$ is unique given $\ell$ and $x$, but if $x \in \ell$ then $\pi$ is a uniform plane containing $\ell$) and output $(x,\ell,\pi)$. Now consider the related distribution $D_2$ obtained by sampling a plane $\pi$ uniformly in $\F^m$, then sampling a line $\ell$ uniformly in $\pi$ and a point $x\in S$ on $\pi$ and outputting $(x,\ell,\pi)$. If there is no point of $S$ in $\pi$ (i.e., $S \cap \pi = \emptyset$), return $\bot$. The following claim establishes that these two distributions are $O(1/(\epsilon  q))$-close in statistical distance. 

\newcommand{\Donep}{\widetilde{D_1}}
\newcommand{\Dtwop}{\widetilde{D_2}}

\begin{claim}\label{clm:bootstrap-tvd}
    Suppose $|\F| \geq \sfrac{3}{\gamma\cdot \sqrt{\mu}}$ where $\mu= \Pr_{x}[x \in S]$. Then $\|D_1 - D_2\|_{\mathrm{TV}} \leq \gamma$
\end{claim}
\begin{proof}
We construct two related distributions $\Donep$ and $\Dtwop$ as follows. $\Donep$ is the distribution obtained by picking $x$ to be a uniform point in $S$, $\ell$ a uniformly random line in $\F^m$ not containing $x$ in $\F^m$ and $\pi$ be a uniform plane containing $\ell$ and $x$ (note $\pi$ is unique given $\ell$ and $x$) and output $(x,\ell,\pi)$. The distribution $\Dtwop$ is obtained by sampling a plane $\pi$ uniformly in $\F^m$, a point $x\in S$ on $\pi$, a line $\ell$ uniformly in $\pi$ not containing $x$ and then outputting $(x,\ell,\pi)$. If there is no point of $S$ in $\pi$ (i.e., $S \cap \pi = \emptyset$), return $\bot$. The only difference between $D1$ and $\Donep$ is that the line $\ell$ definitely does not pass through $x$ in $\Donep$ while it may do so with probability with probability $\sfrac{1}{q^{m-1}}$ in $D_1$ . Thus, $\|D_1-\Donep\|_{\mathrm{TV}}\leq \sfrac{1}{q^{m-1}}$. Similarly, $\|D_2 - \Dtwop\|_{\mathrm{TV}} \leq \sfrac{1}{q}$. Hence, $\|D_1 - D_2\|_{\mathrm{TV}} \leq \|\Donep - \Dtwop\|_{\mathrm{TV}} + \sfrac{1}{q} +\sfrac{1}{q^{m-1}}\leq \sfrac{2}{q}$.

Consider the distribution $\Donep$. It can alternatively be sampled as follows. Pick a random point $x \in S$, a random plane $\pi$ containing $x$, a line $\ell$ uniformly in $\pi$ not containing $x$ and then outputting the triple $(x,\ell,\pi)$. Since the generative processes for picking $\ell$ in $\Donep$ and $\Dtwop$ are identical given the pair $(x,\pi)$ which are picked differently, the distance between the distributions $\Donep$ and $\Dtwop$ is exactly the distance between the marginals of $\Donep$ and $\Dtwop$ on the the $(x,\pi)$ coordinates. By \cref{lem:bipartite-eml-withbadset}, this distance is at most $\sfrac{\lambda}{\mu(S)}$ where $\lambda$ is the second eigen-value of the points-plane incidence graph $G=G(\F^m,\calP^{(m)})$ which is at most $\sfrac{1}{q}$ (\cref{thm:spectralgraphs}-\ref{spec:point-plane-graph}). Hence, $\|\Donep-\Dtwop\|_{\mathrm{TV}}\leq \sfrac{1}{q\sqrt{\mu}}$. 

Hence, $\|D_1-D_2\|_{\mathrm{TV}} \leq \sfrac{2}{q}+\sfrac{q}{q\sqrt{\mu}}\leq \gamma$ provided $q \geq \sfrac{3}{\gamma\cdot \sqrt{\mu}}$.
\end{proof}  

Given this claim, it suffices to prove the following lemma to prove \cref{clm:randomfcorrxpassesldt}.

\begin{lemma}\label{lem:randomfcorrxpassesldt}
  Let the field $\F$, degree $d$ and $\mu, \gamma, \eps \in (0,1)$ satisfy $q \geq 800d\cdot \max\left\{\nicefrac{1}{(\mu\gamma)^{16}}, \nicefrac{1}{\mu\gamma^2\eps^{16}}\right\}$. Given a function $f$ and $\eps$-well-behaved degree-$d$ lines oracle, let $S$ be a subset of the $2\eps$-good points in $\F^m$ of density $\mu$. Then, the distribution $D_2$ on triples $(\pi, x,\ell)$ obtained by picking a random plane $\pi$ in $\F^m$, a random $x\in S$ in the plane $\pi$ and a random line $\ell$ in the plane $\pi$satisfies
    \[ \Pr_{(\pi,x,\ell)\sim D_2}\left[ \exists \text{ degree-$d$ polynomial $P_\ell^x$ such that } \delta(f^x_{\corr_{{\eps^8}/{2}}}|_\ell, P^x_\ell) \leq \gamma \right] \geq 1- \gamma \,.\]  
\end{lemma}    

\begin{proof}[Proof of \cref{clm:randomfcorrxpassesldt}]
Follows from \cref{clm:bootstrap-tvd} and \cref{lem:randomfcorrxpassesldt}
\end{proof}   
\begin{proof}
To prove \cref{lem:randomfcorrxpassesldt}, we list some bad events E1--E4 such that if none of these occur then $f_{\corr_{{\eps^8}/{2}}}^x$ passes the low-degree test on $(\ell,y)$ with probability $1-2\gamma$ over the choice of $y \in \ell$. We then argue that each of these events happens with probability at most $\gamma/4$. 

\begin{description}

    \item[Event E1$(\pi)$:] $\Pr_{z\in \pi}[z \in S] <\nicefrac{\mu}{2}$.
    
    \item[Event E2$(\pi)$:] $\Pr_{z \in \pi}[z \in S \text{ but $z$ is not ${\eps}$-locally-good wrt. plane } \pi] \geq \nicefrac{\mu\cdot \gamma}{16}$.

    \item[Event E3$(\pi,x)$:] $\neg (\text{E1} \vee \text{E2})$ and ``$x$ is not $\eps^8$-explained on $\pi$''.
    
    Now if the plane $\pi$ and point $x$ are such that E1, E2 and E3 do not occur, then there exists a bivariate degree-$d$ polynomial $Q_\pi^x$ on the plane $\pi$ such that $\Pr_{z\in \pi}[Q_\pi^x(z)=f(x)]\geq \eps^8$ and $f(x)=Q_\pi^x(x)$. 
    
    \item[Event E4$(\pi,x,\ell)$:] $\neg (\text{E1} \vee \text{E2}\vee \text{E3})$ and $\Pr_{y \in \ell}[f_{\corr_{{\eps^8}/{2}}}^{x}(y) \neq Q_\pi^x(y)] > \gamma$. 
\end{description}
\cref{clm:E1,clm:E2,clm:E3,clm:E4} imply that each of these events occurs with probability at most $\nicefrac{\gamma}{4}$. Furthermore, clearly if none of the 4 events occur, then the polynomial $Q_\pi^x$ disagrees with the corrected function $f_{\corr_{{\eps^8}/{2}}}^x$ on at most $\gamma$-fraction of the points on the line $\ell$. This completes the proof of the lemma assuming \cref{clm:E1,clm:E2,clm:E3,clm:E4}.
\end{proof}

The rest of this section is devoted towards bounding the probability of the events E1--E4. 

\begin{claim}[Bounding E1]\label{clm:E1}
    If $\mu=\Pr_x[ x \in S ]$ and $|\F| \geq \nicefrac{4}{\sqrt{\mu \cdot\gamma}}$, then \[\Pr_\pi[\mathrm{E1}] = \Pr_{\pi}[ |S\cap \pi| < \nicefrac{\mu}{2}\cdot q^2] \leq \nicefrac{\gamma}{4}.\]
\end{claim}
\begin{proof}
    Consider the bipartite points-planes incidence graph $G(\F^{m},\calP^{(m)})$ in $\F^m$ which has second eigen-value at most $\nicefrac{1}{{q}}$ (by \cref{thm:spectralgraphs}-\ref{spec:point-plane-graph}). Consider the functions $g \colon \F^{m} \to \R$ and $h\colon \calP^{(m)} \to \R$ defined as follows: $g(z) := \mathbbm{1}[z \in S]$ and $h(\pi) := \mathbbm{1}[|S \cap \pi| < \nicefrac{\mu}{2}\cdot q^2]$. These functions satisfy $\mu_g = \mu$, $\mu_h = \Pr[\text{E1}] =: \alpha$, $\sigma^2_g \leq  \mu$ and $\sigma^2_h \leq \alpha$. Applying \cref{lem:bipartite-eml} to the graph $G$ with functions $g$ and $h$ as defined above, we have
    \[ \mu \cdot \alpha - \nicefrac{\mu}{2}\cdot \alpha \leq \lambda \cdot \sqrt{\mu \cdot \alpha}. \]
    Equivalently, $\alpha \leq\nicefrac{4\lambda^2}{\mu}$, which in turn is at most $\nicefrac{\gamma}{4}$ (since $\lambda \approx \nicefrac{1}{|\F|}$ and $|\F| \geq \nicefrac{4}{\sqrt{\mu\cdot\gamma}}$).
\end{proof}

\begin{claim}[Bounding E2]\label{clm:E2}
    If $|\F| +1\geq \nicefrac{100}{{\eps\cdot\gamma^2 }}$ and $\mu=\Pr_x[ x \in S]$, then
    \[ \Pr_\pi[\mathrm{E2}] = \Pr_\pi[\pi \text{ has at least } \nicefrac{\mu\cdot \gamma}{16}\cdot q^2 \text{points which are in $S$ but not ${\eps}$-locally-good }] \leq \nicefrac{\gamma}{4}. \]
\end{claim} 
\begin{proof}
    We begin by showing the following: If for any $2\eps$-good point $x$ (in particular, if $x\in S$) and $|\F| +1\geq \nicefrac{100}{{\eps\cdot\gamma^2}}$, then 
    \begin{equation}\label{eq:globalbutnotlocal}
        \Pr_{\pi \in \calP_x}[ x \text{ is not ${\eps}$-good with respect to plane } \pi] \leq \nicefrac{\gamma^2}{64}.
    \end{equation}

    Let $x$ be an $2\eps$-good point. Let $L_x \subseteq \calL^{(m)}_x$ be the set of lines $\ell$ containing $x$ that satisfy $P_\ell(x) = f(x)$. We have $\Pr_{\ell \ni x}[\ell \in L_x] \geq 2\eps$. Consider the bipartite lines-plane incidence graph $G(\calL^{(m)}_x,\calP^{(m)}_x)$ in $\F^m$ between lines and planes containing $x$. This graph has second eigen-value at most $\nicefrac{1}{\sqrt{q+1}}$ (by \cref{thm:spectralgraphs}-\ref{spec:lines-planes-graph-containing-x}). Consider the functions $g \colon \calL^{(m)}_x \to \R$ and $h\colon \calP^{(m)}_x \to \R$ defined as follows: $g(\ell) := \mathbbm{1}[f(x) = P_\ell(x)] = \mathbbm{1}[\ell \in L_x]$ and $h(\pi) := \mathbbm{1}[\Pr_{\ell \in \pi}[\ell \in L_x] < {\eps}]$. These functions satisfy $\mu_h =:\alpha$, $\sigma^2_g \leq  \mu_g$ and $\sigma^2_h \leq \alpha$. Applying \cref{lem:bipartite-eml} to the graph $G$ with functions $g$ and $h$ as defined above, we have
    \[ \mu_g \cdot \alpha - {\eps}\cdot \alpha \leq \lambda \cdot \sqrt{\mu_g \cdot \alpha}. \]
    Equivalently, $\alpha \leq\nicefrac{\lambda^2\cdot \mu_g}{(\mu_g-{\eps})^2}$ which is at most $\nicefrac{\lambda^2}{\eps}$ (since $\mu_g \geq 2\eps$), which in turn is at most $\nicefrac{\gamma^2}{64}$ (since $\lambda \approx \nicefrac{1}{\sqrt{|\F|+1}}$ and $|\F| +1\geq \nicefrac{64}{{\eps\cdot\gamma^2}}$), completing the proof of \eqref{eq:globalbutnotlocal}.

    We now return to the proof of \cref{clm:E2}. Consider the bipartite point-plane inclusion graph $G=G(\F^m, \calP^{(m)})$. Let $A :=S$ and $\mu:= \Pr_z[z \in A]$. Let $B \subset \calP^{(m)}$ be the set of planes $\pi$ which have at least $\nicefrac{\mu\cdot \gamma}{16}\cdot q^2$ points which are in $S$ but not ${\eps}$-locally-good. We mark an edge $(x,\pi)$ if $x \in A$ (i.e., $x$ is $2\eps$-globally-good) but $x$ is not ${\eps}$-good with respect to the plane $\pi$. The fraction of marked edges in $G$ by \cref{eq:globalbutnotlocal} is at most $\mu \cdot \nicefrac{\gamma^2}{64}$. On the other hand, the fraction of marked edges is at least $\mu(B)\cdot \nicefrac{\mu\cdot\gamma}{16}$. Hence, $\mu(B)\cdot \nicefrac{\mu\cdot\gamma}{16} \leq \mu \cdot \nicefrac{\gamma^2}{64}$ or equivalently, $\mu(B) \leq \nicefrac{\gamma}{4}$.
\end{proof}

\begin{claim}[Bounding E3]\label{clm:E3}
    If the parameter $\eps_0$ in \eqref{eq:strong-bi-ldt} satisfies $\eps_0 \leq \nicefrac{\mu\cdot \gamma}{16}$ where $\mu= \Pr_z[z \in S]$, then
    $\Pr_{\pi,x}[\mathrm{E3}] \leq \nicefrac{\gamma}{4}$.
\end{claim}
\begin{proof}
    To begin with fix a plane $\pi$ such that neither E1 nor E2 occurs. Hence, $\pi$ has at least $\nicefrac{\mu}{2}\cdot q^2$ points in $S$  of which at most $\nicefrac{\mu\cdot\gamma}{16}\cdot q^2$ are not $\eps$-locally-good. Hence, there are at least $\nicefrac{\mu\cdot(1-\nicefrac{\gamma}{8})}{2}\cdot q^2 \geq \nicefrac{\mu}{4}\cdot q^2$ points which are $\eps$-locally-good in $\pi$. 

    Now, by the bivariate LDT \eqref{eq:strong-bi-ldt}, we have that the probability that a random $x \in \F^m$ is $\eps$-locally-good but not $\eps^8$-explained (both with respect to $\pi)$ is at most $\eps_0$. Hence, the probability that a random point in $S\cap \pi$ is not $\eps^8$-explained is at most 
    \[ \frac{\nicefrac{\mu\cdot\gamma}{16}+ \eps_0}{\mu\left(S\cap \pi \right)} \leq \frac{\nicefrac{\mu\cdot\gamma}{16}+ \eps_0}{\nicefrac{\mu}{2}}= \nicefrac{\gamma}{8} + \nicefrac{2\eps_0}{\mu} \leq \nicefrac{\gamma}{4},\] 
    where the last inequality follows if $\eps_0 \leq \nicefrac{\mu\cdot\gamma}{16}$. We now bound E3 as follows:
    \[\Pr_{\pi,x}[\mathrm{E3}] \leq \E_\pi[\Pr_{x}[\mathrm{E3} | \neg (\mathrm{E1} \vee \mathrm{E2})]] = \E_\pi[\Pr_{x\in \pi}[ x \text{ is not $\eps^8$-explained }| \neg (\mathrm{E1} \vee \mathrm{E2}) ]] \leq \nicefrac{\gamma}{4}. \qedhere\]
\end{proof}

\begin{claim}[Bounding E4]\label{clm:E4}
  If $|\F| \geq \nicefrac{800d}{\mu\cdot\gamma^2\cdot \eps^{16}}$, then $\Pr_{\pi,x,\ell}[\mathrm{E4}] \leq \nicefrac{\gamma}{4}$.x
\end{claim}
\begin{proof}
    The argument for bounding $\Pr[\text{E4}]$ will be far more involved and subtler than the previous cases.

    Let us for the rest of the argument fix a plane $\pi$. Furthermore, let us assume that this plane $\pi$ and a random $x\in S$ chosen on it are such that none of E1, E2 and E3 occur. This implies that $\pi$ has at least $\nicefrac{\mu}{2}$-fraction of points which are in $S$ and there exists a bivariate degree-$d$ polynomial $Q_\pi^x$ such that (1) $\Pr_{z\in \pi}[Q_\pi^x(z)=f(z)] \geq \eps^8$ and (2) $Q_\pi^x(x)=f(x)$. We now need to bound the probability that when we choose a random line $\ell$ in the plane $\pi$, at least $\gamma$-fraction of the points $y$ on $\ell$ satisfy $Q_\pi^x(y)\neq f_{\corr_{{\eps^8}/{2}}}^x(y)$. To this end, let us recall the definition of $f_{\corr_{{\eps^8}/{2}}}^x \colon \F^m \to \F \cup \{\bot\}$. For any point $y\in \F^m \setminus \{x\}$, let $P$ be the unique degree-$d$ univariate polynomial $P$ on the line $\ell'=\ell_{x,y}$ through $x$ and $y$ such that (1) $\Pr_{z\in \ell'}[P(z) =f(z)]\geq {{\eps^8}/{2}}$ and (2) $P(x)= f(x)$. If there is no such polynomial $P$ or there is more than one such polynomial, then we set $P$ to $\bot$. Finally, we set $f_{\corr_{{\eps^8}/{2}}}^x(y):=P(y)$ and $f_{\corr_{{\eps^8}/{2}}}^x(x):=f(x)$. We will now argue that the probability that this polynomial $P$ is $Q_\pi^x|_{\ell'}$ for at least $1-\gamma$ fraction of the points $y$ in $\ell$ (this probability will be over the random choice of an $\eps$-good point $x$ in $\pi$ and a random line $\ell$ in $\pi$).

  Let $Q_1, \ldots, Q_t$ be the list of all degree-$d$ bivariate polynomials on $\pi$ such that $\Pr_{y \in \pi}[Q_j(y)=f(y)]\geq \eps^8$. By Johnson bound \cref{thm:johnson}-\ref{item:johnson}, we have $t \leq \nicefrac{2}{\eps^8}$ provided $\eps^8 \geq 2\sqrt{\nicefrac{d}{q}}$. Note that $Q_\pi^x$ is one such polynomial. For $j \in [t]$, let $S_j := \{ x \in \pi \colon Q_j(x)=f(x)\}$ be the set of agreement points between $f$ and the polynomial $Q_j$ (by definition, $\mu(S_j)\geq \eps^8$). 
    
    We list below some bad events B1, B2 and B3, which if they do not occur would imply that for a random $x\in S$, a random line $\ell$ and a point $y \in \ell$, the polynomial $P$ is the restricted polyomial $Q_\pi^x|_{\ell'}$ (where $\ell'=\ell_{x,y}$). These bad events would be described over the randomness of the choice of the point $x\in S$ and the random line $\ell' \in \calL_x^{\pi}$ (note $\ell' = \ell_{x,y}$ is a random line through the point $x$ in the plane $\pi$)

    \begin{description}
        \item[Event B1$(\ell')$:] There exists $j \in [t]$ such that $|S_j\cap \ell' | < \nicefrac{\eps^8}{2}\cdot q$.
        \item[Event B2$(\ell')$:] $|S \cap \ell'| < \nicefrac{\mu}{4}\cdot q$.
        
        Let $P_1,\ldots, P_r$ be the list of all univariate degree-$d$ polynomials on the line $\ell'$ such that $\Pr_{z\in \ell'}[P_k(z)=f(z)] \geq \nicefrac{\eps^8}{2}$. Note if B1 does not occur, then $Q_\pi^x|_{\ell'}$ is one such polynomial. By Johnson bound \cref{thm:johnson}, we have $r \leq \nicefrac{4}{\eps^8}$ provided $\nicefrac{\eps^8}{2} \geq 2\sqrt{\nicefrac{d}{q}}$. We say $x$ is a non-unique point on $\ell'$ if there exist two distinct polynomials $k \neq k' \in [r]$ such that $P_k(x) = P_{k'}(x)$. 

        \item[Event B3$(\ell',x)$:] $\neg$B2 and $x$ is a non-unique point on $\ell'$. 
    \end{description} 
    Applying \cref{clm:b123} with $\delta := \nicefrac{\gamma^2}{8}$, we have $\Pr_{(\ell',x)\sim \calE}[\mathrm{B1}\vee \mathrm{B2}\vee \mathrm{B3}] \leq \nicefrac{\gamma^2}{8}$ provided $q \geq \nicefrac{800 d}{\mu \cdot \gamma^2\cdot \eps^{16}}$. However, the distribution $\calE$ over the pairs $(x,\ell')$ is slightly different from what we need. Our required distribution of $(x,\ell')$ is as follows: given a plane $\pi$, pick a random  point $x\in S$ on it and a random line $\ell'$ in $\pi$ passing through $x$. The distribution $\calE$, on the other hand, is as follows: given a plane $\pi$, pick a random line $\ell'$ and a random $x\in S$ on $\ell'$ (if one exists). By \cref{lem:bipartite-eml-withbadset}, these two distributions are $\nicefrac{\lambda}{\sqrt{\mu\left(S\cap\pi\right)}}$-close in total variation distance where $\lambda\approx \nicefrac{1}{\sqrt{q}}$ is the second eigen-value of the points-line incidence graph in the plane $\pi$ (\cref{lem:bipartite-eml-withbadset}-\ref{spec:point-line-in-plane-graph}) and $\mu\left(S\cap\pi\right)\geq \nicefrac{\mu}{2}$ (since event E1 does not hold). Hence, this distance is at most $\sqrt{\nicefrac{2}{q\cdot \mu}}$ which is at most $\nicefrac{\gamma^2}{8}$ provided $q \geq \nicefrac{200}{\mu\cdot\gamma^4}$. Hence, $\Pr_{x,\ell'}[\mathrm{B1}\vee \mathrm{B2}\vee \mathrm{B3}]\leq \nicefrac{\gamma^2}{4}$. 
    
    Now, recall that $\pi$ and $x$ are such that none of the events E1--E3 occur. Hence, $\pi$ has at least $\nicefrac{\mu}{2}$-fraction of $2\eps$-good points and there exists a bivariate degree-$d$ polynomial $Q_\pi^x$ such that (1) $\Pr_{z\in \pi}[Q_\pi^x(z)=f(z)] \geq \eps^8$ and (2) $Q_\pi^x(x)=f(x)$. Now, suppose furthermore that $\ell$ and $y \in \ell$ (and the corresponding $\ell'=\ell_{x,y}$) are such that B1--B3 do not occur. It follows from these assumptions that $Q_\pi^x$ is one of the polynomials $Q_j$ and $Q_\pi^x|_{\ell'}$ one of the polynomials $P_k$. The uniqueness condition of $\neg$B3 implies that $Q_\pi^x|_{\ell'}$ is the only degree-$d$ polynomial that has agreement at least $\nicefrac{\eps^8}{2}$ with $f$ on $\ell$ and $Q_\pi^x(x)=f(x)$. Hence, $f_{\corr_{\nicefrac{\eps^8}{2}}}^x(y)=Q_\pi^x(y)$. We have thus shown that
    \[ \Pr_{x,\ell, y \in \ell}[f_{\corr_{\nicefrac{\eps^8}{2}}}^x(y)\neq Q_\pi^x(y)]\leq \nicefrac{\gamma^2}{4}.\] 
    We are however interested in the fraction of points in $\ell$ such that $f_{\corr_{\nicefrac{\eps^8}{2}}}^x(y)\neq Q_\pi^x(y)$. A Markov argument shows that 
    \[ \Pr_{x,\ell}\left[ \left|\{ y \in \ell \colon f_{\corr_{\nicefrac{\eps^8}{2}}}^x(y)\neq Q_\pi^x(y)\}\right| > \gamma \cdot  q\right] \leq \nicefrac{\gamma^2\cdot q}{4\cdot \gamma\cdot q} = \nicefrac{\gamma}{4}. \] 
    Averaging over $\pi$, yields the claim.
\end{proof}

\begin{claim}\label{clm:b123}
    Let $\rho,\delta,\eps \in (0,1)$, field $\F$ of size $q$ and degree parameter $d$ satisfy $q \geq \nicefrac{100d}{\mu\cdot \delta\cdot \eps^{16}}$. For any plane $\pi$ and a set $S$ of $2\eps$-good points in the plane of density at least $\nicefrac{\mu}{2}$, consider the distribution $\calE$ on pairs $(\ell',x)$ chosen as follows: $\ell'$ is a random line in the plane $\pi$ and $x$ is a random point on $S\cap\ell$ (if not such point exists, then the distribution returns $\bot$). Then, $\Pr_{(\ell,x')\sim \calE}[\mathrm{B1}\vee \mathrm{B2}\vee \mathrm{B3}] \leq \delta$.
\end{claim}
\begin{proof}
    We bound the probability of each of the events B1, B2 and B3 by $\nicefrac{\delta}{3}$ as follows:
    \begin{description}
        \item[Event B1$(\ell')$:] There exists $j \in [t]$ such that $|S_j\cap \ell'| < \nicefrac{\eps^8}{2}\cdot q$.
        
        For any fixed $j \in [t]$, since $\mu(S_j) \geq \eps^8$ and the set of points in a random line $\ell'$ are pairwise independent, we have $\Pr[ |S_j\cap \ell'| < \nicefrac{\eps^8}{2}\cdot q] \leq \nicefrac{q\cdot \eps^8}{(q \cdot\eps^8/2)^2}=\nicefrac{4}{q\eps^8}$. Hence, $\Pr[\text{B1}] \leq t \cdot \nicefrac{4}{q\eps^8} \leq \nicefrac{8}{q\cdot \eps^{16}}$ (as $t \leq \nicefrac{2}{\eps^8}$). We, hence, have $\Pr[\text{B1}]\leq \nicefrac{\delta}{3}$ provided $q \geq \nicefrac{24}{\delta \cdot \eps^{16}}$
        
        \item[Event B2$(\ell')$:] $|S\cap \ell'| < \nicefrac{\mu}{4}\cdot q$. 
        
        Since $\mu\left(S\cap \pi\right) \geq \nicefrac{\mu}{2}$ and the set of points in a random line $\ell'$ are pairwise independent, we have $\Pr[\text{B2}]=\Pr[ |S \cap \ell'| < \nicefrac{\mu}{4}\cdot q] \leq \nicefrac{q\cdot \mu/2}{(q \cdot\mu/4)^2}=\nicefrac{8}{q\mu}$, which is at most $\nicefrac{\delta}{3}$ provided $q \geq \nicefrac{24}{\mu\cdot \delta}$. 

        \item[Event B3$(\ell',x)$:] $\neg$B2 and $x$ is a non-unique point on $\ell'$. 
        
        Since event B2 does not happen, we know that that at least $\nicefrac{\mu}{4}$-fraction of points in $\ell'$ are $2\eps$-good. The fraction of points $x$ in $\ell$ such that there exist two distinct polynomials $P_k$ $P_k'$ which coincide on $x$ (i.e., $P_k(x)=P_{k'}(x))$ is at most $\binom{r}{2}\cdot \nicefrac{d}{q} \leq \nicefrac{8d}{q\cdot (\eps^8)^2}$ since $r \leq \nicefrac{4}{\eps^8}$. Hence, the probability of event B3 that a random point in $\ell'\cap S$ happens to be non-unique is at most $\nicefrac{4\cdot 8d}{\mu\cdot q\cdot (\eps^8)^2}$. Hence, $\Pr[\text{B3}] \leq \nicefrac{\delta}{3}$ provided $q \geq \nicefrac{100d}{\mu\cdot\delta\cdot \eps^{16}}$.  \qedhere
    \end{description}
\end{proof}

\section*{Acknowledgements}
Some of the discussions of the first two authors leading up to this work happened while they were visiting the Homi Bhabha Center for Science Education (HBCSE), Mumbai. We are thankful to Prof. Arnab Bhattacharya and rest of the HBCSE staff for their warm and generous hospitality and for providing an inviting and conducive atmosphere for these discussions.

{\small 
  \bibliographystyle{prahladhurl}
  \bibliography{ldt-bib.bib}
}
{\let\thefootnote\relax
\footnotetext{\textcolor{\gitinfonotecolour}{\gitinfonote \easteregg}
}}

\appendix

\section{An Arora-Sudan style interpolation}\label{sec:bivariate-interpolation}

In this section, we present the full proof of \cref{lem:interpolation-step}. We repeat the statement below for convenience. 

\interpolationstep*

The proof is mostly along the lines of Arora and Sudan~\cite{AroraS2003} but a tighter analysis and some care for fields of small characteristic. We present it in its entirety for completeness and to make the changes clearer to follow. 

\subsection{Finding a structured subset of points}
We start with the following claim.
\begin{claim}[Finding good directions]
  \label{claim:AS-good-directions}
  Let $q > \sfrac{8}{\epsilon}$ and suppose $f:\F^2 \rightarrow \F$ passes $\LDT_d$ with probability $\epsilon$. Then, there are two different directions $\ell_1$ and $\ell_2$, and a set $H \subseteq \F^2$ such that
  \begin{itemize}\itemsep 0pt
    \item $|H| = \Omega(\epsilon^2 \cdot q^2)$,
    \item for all $x \in H$ we have $\Pr_{\ell \ni x}[P_{\ell}(x) = f(x)] \geq \sfrac{\epsilon}{2}$ and 
    \[
      P_{\ell_{1}^{(x)}}(x) = P_{\ell_{2}^{(x)}}(x) = f(x)
    \]
    where $\ell_{i}^{(x)}$ is the line through $x$ parallel to $\ell_i$.
  \end{itemize}
\end{claim}
\begin{proof}
  For each $x \in \F^2$, let $\epsilon_x = \Pr_{\ell\ni x}[P_\ell(x) = f(x)]$ and we have $\sum_x \epsilon_x = \epsilon \cdot q^2$. Let $H' = \setdef{x}{\epsilon_x \geq \sfrac{\epsilon}{2}}$. For a direction $\ell$ and a point $x \in \F^2$, let $I(x,\ell)$ be indicator random variable defined as
  \[
    I(x,\ell) = \mathbbm{1}\insquare{P_{\ell^{(x)}}(x) = f(x) \;\text{and}\; \epsilon_x \geq \sfrac{\epsilon}{2}}
  \] 
  where $\ell^{(x)}$ is the line through $x$ parallel to $\ell$. We then have $\E_\ell[I(x,\ell)] = \epsilon_x$ if $x\in H'$ and $0$ otherwise. Thus, 
  \[
    \E_x \E_\ell [I(x,\ell)] = \sfrac{1}{q^2} \cdot \sum_{x\in H'} \epsilon_x \geq \sfrac{\epsilon}{2} \quad\text{(by \cref{lem:averaging})}.
  \]
  Therefore, we have
  \begin{align*}
    \sfrac{\epsilon}{2} & \leq \E_{x}\E_\ell[I(x,\ell)]\\
    \sfrac{\epsilon^2}{4} & \leq \inparen{\E_{x}\E_\ell[I(x,\ell)]}^2 \leq \E_{x}(\E_\ell[I(x,\ell)])^2\\
    & =  \E_{x}\inparen{\E_{\ell_1,\ell_2}[I(x,\ell_1) \cdot I(x,\ell_2)]}\\
    & = \E_{x}\inparen{\E_{\ell_1\neq \ell_2}[I(x,\ell_1) \cdot I(x,\ell_2)] + \frac{1}{q^2 + q}} & \text{($\because \; \Pr[\ell = \ell'] = \frac{1}{q^2 + q}$)}\\
    \implies \sfrac{\epsilon^2}{8} & \leq \E_{x}\E_{\ell_1\neq \ell_2}[I(x,\ell_1) \cdot I(x,\ell_2)] & \text{($\because \;  q > \sfrac{8}{\epsilon}$)}
  \end{align*}
  Therefore, there exist two different directions $\ell_1$ and $\ell_2$ such that $\E_{x}[I(x,\ell_1) I(x,\ell_2)] \geq \sfrac{\epsilon^2}{8}$. Fixing such directions $\ell_1$ and $\ell_2$, defining $H = \setdef{x}{I(x,\ell_1) = I(x,\ell_2) = 1}$ satisfies the requirements.
\end{proof}

Without of loss of generality, we assume that $\ell_1$ is the $x$-axis and $\ell_2$ is the $y$-axis. 
Let $\gamma$ be chosen such that $H = 2\gamma \cdot q^2 = \Omega(\epsilon^2 q^2)$. 

\begin{lemma}[Structured subsets within $H$]
  \label{lem:structured-grid}
  Let $H$ be the set specified above with respect to the $x$ and $y$ directions for an oracle $f$ passing $\LDT_d$ with probability at least $\epsilon$. Then, for any $r$ satisfying $\frac{2\log q}{\gamma^2} \leq r \leq \gamma \cdot q$, there are subsets $S_1, S_2 \subseteq \F$ such that 
  \begin{enumerate}\itemsep0pt
    \item \label{lem:structured-subset-S1-size} $|S_1| = r$,
    \item \label{lem:structured-subset-S2-size} $|S_2| \geq \gamma \cdot q$,
    \item \label{lem:structured-subset-F-times-S2-size} $|(\F \times S_2) \cap H| \geq \sfrac{|H|}{2} = \gamma \cdot q^2$,
    \item \label{lem:structured-subset-H-in-rows} For all $b\in S_2$, we have $\abs{\setdef{(a,b)}{a\in \F} \cap H} \geq \gamma \cdot q$,
    \item \label{lem:structured-subset-H-in-rows-in-S1} For all $b \in S_2$, we have $\abs{\setdef{(a,b)}{a\in S_1} \cap H} \geq \sfrac{\gamma}{2} \cdot \abs{S_1}$.
  \end{enumerate}
\end{lemma}
\begin{proof}
  Let $\ell_{y = b} = \setdef{(a,b)}{a \in \F}$ and let $\ell_{x=a} = \setdef{(a,b)}{b\in \F}$. Since these lines evenly cover the space, we have that $\E_b[\abs{\ell_{y=b} \cap H}] = 2\gamma \cdot q$. Let $S_2 = \setdef{b\in \F}{\abs{\ell_{y=b} \cap H} \geq \gamma q}$. By \cref{lem:averaging}, we have that $|S_2| \geq \gamma \cdot q$ and $\abs{(\F \times S_2) \cap H} \geq \sfrac{|H|}{2} = \gamma \cdot q^2$. 

  Let $S_1$ be a set of $r$ distinct elements of $\F$ chosen uniformly at random. Then, for any $b\in S_2$
  \begin{align*}
    \E_{S_1}[ \abs{\setdef{(a,b)}{a\in S_1\;,\;(a,b)\in H}}] & = \sfrac{r}{q} \cdot \abs{\ell_{y=b} \cap H} \geq \gamma \cdot r.
  \end{align*}
  From standard tail bounds for hypergeometric distributions (cf. \cite{Skala2013-hypergeometric}), we have 
  \begin{align*}
  \text{For all $b\in S_2$, }\Pr_{S_1}\insquare{\abs{\setdef{(a,b)}{a\in S_1\;,\;(a,b)\in H}}  < \sfrac{\gamma r}{2}} &\leq \exp(-\sfrac{\gamma^2 r}{2})\\
  \implies \Pr_{S_1}\insquare{\exists b\in S_2\;:\;\abs{\setdef{(a,b)}{a\in S_1\;,\;(a,b)\in H}} < \sfrac{\gamma r}{2}} &\leq |S_2| \cdot \exp(-\sfrac{\gamma^2 r}{2}) < 1.
  \end{align*}
  Thus, there exists a set $S_1$ of size $r$ such that for every $b \in S_2$ we have 
  \[
    \abs{\setdef{(a,b)}{a\in S_1\;,\;(a,b)\in H}} \geq \sfrac{\gamma r}{2}.\qedhere
  \]
\end{proof}

\subsection{Properties of the desired interpolating polynomial}

As preparation of the interpolating step, we start by describing the desired monomial support of the interpolating polynomials and observing some of their properties. 
 
 \begin{definition}
For $n, N, p \in \N$ such that $p \geq 2$, the sets $N_{d, D}$ and $N_{d, D, p}$ are defined as follows. 
 
 \begin{itemize}
 \item $N_{d, D}$ is the set of exponent vectors of trivariate monomials with $(1,1,d)$-degree bounded by $D$, i.e.,
 \[
 N_{d,D} := \setdef{(i,j,k) \in \Z_{\geq 0}^3}{i + j + dk \leq D} \, .
\]  
\item $N_{d, D, p}$ is the set of exponent vectors of trivariate monomials $x^iy^jz^k$ with $(1, 1, d)$-weighted degree at most $D$ such that the $z$ degree is either zero, or is not a multiple of $p$. More formally,
\[
N_{d,D, p} := \setdef{(i,j,k) \in \Z_{\geq 0}^3}{i + j + dk \leq D\text{ and } (\text{$k$ is either zero or $p\nmid k$})}.\qedhere
\]
\end{itemize}
\end{definition}
The following simple claims now give us bounds on the sizes of the sets $N_{d, D}$ and $N_{d, D, p}$. We defer the proofs of the claims to \cref{sec:proofs-of-estimates-on-Ndd-etc}.
\begin{claim}\label{clm:NdD-lower-bound} 
For all $d, D \in \N$, we have 
\[
(D^3/3d - 5/2D^2 + Dd/6) \leq |N_{d, D}| \leq (D^3/3d + 3/2D^2 + Dd/6) \, .
\]
\end{claim}
\begin{claim}\label{clm:NdDp-lower-bound}
For all $d, D \in \N$, with $p \geq 2$, we have  $|N_{d, D, p}| \geq |N_{d, D}|/2$. 
\end{claim}
Finally, we combine the bounds in \cref{clm:NdD-lower-bound} and \cref{clm:NdDp-lower-bound} to get the following claim. 
\begin{claim}\label{clm:NdDp-lowerbound-final}
For all $d, D, p \in \N$ such  that $D > 20d$ and $p \geq 2$, $|N_{d, D, p}| \geq \frac{D^3}{12d}$. 
\end{claim}

\subsection*{Interpolation}

We now rely on the above estimates to interpolate a low-degree polynomial that explains the function value in the points table on a constant fraction of points. But first, we recall some notation: $H, S_1, S_2$ be as guaranteed by \cref{lem:structured-grid} and let $r$ be the size of $S_1$. Define $S = \setdef{(a,b) \in \F^2}{b \in S_2\;,\;(a,b) \in H}$. From \cref{lem:structured-grid}, we have that $|S| \geq \gamma \cdot q^2$. With this notation in place, we have the following lemma. 
\begin{lemma}
  \label{lem:AS-finding-vanishing-polynomial}
Let $\F$ be a finite field of characteristic $p$ and let $D$ be a natural number satisfying $\abs{N_{d,D, p}} >r (D+1)$. There exists a non-zero polynomial $A(x,y,z) \in \F[x, y, z]$ with $\deg_{1,1,d}(A) \leq D$ such that 
  
\begin{itemize}
\item For every $a \in S_1$, we have
  \[
    A(a, y, P_{\ell_{x = a}}(y)) = 0
  \]
  as a univariate polynomial in $y$, where $P_{\ell_{x=a}}$ is the best-fit degree-$d$ polynomial for $f$ the line $\ell_{x = a}$. 
\item If $p$ is the characteristic of the field $\F$, then, $A$ is supported on monomials of the form $x^iy^jz^k$ where $k$ is either zero or is not divisible by $p$. 
\end{itemize}  
    
\end{lemma}
\begin{proof}
  Let $A(x,y,z) = \sum_{i,j,k \in N_{d,D}} A_{i,j,k} x^i y^j z^k$, where $A_{i,j,k}$ are indeterminates. For each $a \in S_1$, the constraint $A(a,y,P_{\ell_{x=a}}(y)) = 0$ can be viewed as $(D+1)$ homogeneous linear constraints on the indeterminates $A_{i,j,k}$ obtained by equating the coefficients of $y^0, y^1, \ldots, y^{D}$ to zero. Note that since the $(1, 1, d)$-weighted degree of every monomial in $A$ is at most $D$, the degree of $A(a,y,P_{\ell_{x=a}}(y))$ is at most $D$. Hence, as long as $\abs{N_{d,D, p}} >  r \cdot (D+1)$, there must exist a non-zero solution to the system of equations, and therefore a non-zero polynomial $(1, 1, d)$-degree $D$ that is entirely supported on monomials from the set $N_{d, D, p}$ and therefore satisfies the second item in the lemma by definition.

\end{proof}
We now infer that any polynomial satisfying the conditions in \cref{lem:AS-finding-vanishing-polynomial} imply that it satisfies many more vanishing conditions.
\begin{lemma}
  \label{lem:AS-A-vanishes-on-more-points}
Suppose $A(x,y,z)$ is a polynomial with $\deg_{1,1,d}(A) \leq D$ such that for all $a\in S_1$ we have $A(a, y, P_{\ell_{x = a}}(y)) = 0$. If $D < \sfrac{\gamma}{2} \cdot |S_1|$, then we have that
\[
  A(a,b,f(a,b)) = 0 \quad\text{for all }(a,b) \in S = \setdef{(a,b)}{b \in S_2\;,\; (a,b)\in H}. 
\]
\end{lemma}
\begin{proof}
  Fix a $b \in S_2$ and let $R_b(x) = P_{\ell_{y=b}}$ and $Q(x) = A(x,b,R_b(x))$. Suppose $a \in S_1$ with $(a,b) \in H$, then $R_b(a) = P_{\ell_{x=a}}(b) = f(a,b)$ and hence
  \begin{align*}
    A(a,b,R_b(a)) & = A(a,b,P_{\ell_{x=a}}(a)) = 0.
  \end{align*}
  Since $\deg(Q) \leq \deg_{1,1,d}(A) \leq D$, and $\abs{\setdef{(a,b)}{a\in S_1} \cap H} \geq \sfrac{\gamma}{2} \cdot |S_1| = \sfrac{\gamma r}{2}$, the condition that $D < \sfrac{\gamma r}{2}$ implies that $Q(x)$ is identically zero. Since $Q(x) = A(x,b,R_b(x)) = 0$, for any $a$ such that $(a,b) \in H$, we have that $A(a,b,R_b(a)) = A(a,b,f(a,b)) = 0$ as claimed.
\end{proof}

\subsection{Proof of the interpolation lemma}

We now have the necessary ingredients to prove \cref{lem:interpolation-step}. 

\begin{proof}[Proof of \cref{lem:interpolation-step}]

Let $H$ be the set guaranteed by \cref{claim:AS-good-directions} and $\gamma \in (0, 1]$ such that $H = 2\gamma \cdot q^2 = \Omega(\epsilon^2 \cdot q^2)$; without loss of generality let the two directions guaranteed by \cref{claim:AS-good-directions} be the standard axes. 

Let $r = \sfrac{900 d}{\gamma^2}$ and $D = \sfrac{\gamma r}{3} = \sfrac{300 d}{\gamma}$. From these choice of parameters, we have the following inequalities. 
\begin{itemize}
\item $D > 20d$   
\item $|N_{d, D, p}| \geq D^3/12d = \frac{(300^3 d^2)}{12 \gamma^3}$
\item $r(D+1) \leq 2rD = \frac{6(300)^2 d^2}{\gamma^3}$
\item $\gamma r /2 = \frac{450 d}{\gamma}$
\end{itemize}
Thus, we have that $D \geq 20d$, $|N_{d, D, p}| > r(D+1)$ and $D < \gamma r/2$.

Now, Instantiating \cref{lem:structured-grid} with this choice of $r$, we obtain sets $S_1,S_2$ and let $S$ be defined as 
\[
S = \setdef{(a,b) \in \F^2}{b \in S_2\;,\;(a,b) \in H}\, .
\] 
By construction, $|S| \geq \gamma \cdot q^2 = \Omega(\epsilon^2 \cdot q^2)$. Since $S \subseteq H$ and each $x \in S$ satisfies $\epsilon_x \geq \sfrac{\epsilon}{2}$ (from \cref{claim:AS-good-directions}), we have that $\sum_{x\in S} \epsilon_x = \Omega(\epsilon^3 \cdot q^2)$. 

Since $D^3/12d > r(D+1)$ and $D < \gamma r/2$, we have by \cref{lem:AS-finding-vanishing-polynomial} and \cref{lem:AS-A-vanishes-on-more-points} that there is a non-zero polynomial $\tilde{A}(x,y,z)$ with $\deg_{1,1,d}(\tilde{A}) \leq D$, entirely supported on monomials in the set $N_{d, D, p}$ such that  $\tilde{A}(a,b,f(a,b)) = 0$ for all $(a,b) \in S$. In order to proceed further in the proof, we need the following claim whose proof we defer to the end of this section. 
\begin{claim}\label{clm:A-depends-on-z}
The polynomial $\tilde{A}$ depends on the variable $z$, i.e., there exists a monomial of the form $x^iy^jz^k$ with  non-zero coefficient such that $k$ is non-zero.
\end{claim}
Since $\tilde{A}$ \emph{depends} on $z$ and is only supported on monomials where the $z$ degree is either zero or is not a multiple of $p$, we have from the first item of \cref{prop:hasse-derivatives-properties} and the linearity of Hasse derivatives that $\hpartial_z(\tilde{A})$ must be non-zero. We take ${A}(x, y, z)$ to be a non-zero polynomial of the minimum $(1, 1, d)$-weighted degree such that it vanishes on $(a, b, f(a,b))$ for all $(a,b) \in S$ and $\hpartial_z({A})$ is non-zero. We know that the polynomial $\tilde{A}$ interpolated above is one such polynomial. Therefore, $\deg_{(1,1,d)}({A}) \leq \deg_{(1,1,d)}(\tilde{A}) \leq D$. The polynomial ${A}$ thus obtained satisfies the hypothesis of \cref{lem:minimal-vanishing-polynomial-is-squarefree}, and therefore, by \cref{lem:minimal-vanishing-polynomial-is-squarefree}, we get that $\disc_z({A})$ must be non-zero.

This proves all the properties of $A$ and $S$ claimed by \cref{lem:interpolation-step}.
\end{proof}
To complete the proof of \cref{lem:interpolation-step}, we now prove \cref{clm:A-depends-on-z}.

\subsection*{Proof of \cref{clm:A-depends-on-z}}
\begin{proof}[Proof of \cref{clm:A-depends-on-z}]
If $\tilde{A}$ does not depend on $z$ at all, then it is a non-zero bivariate polynomial of $(1,1)$-weighted degree, i.e., total degree at most $D$, and it vanishes at all points $(a,b)$ in the set $S \subseteq \F \times \F$. But, from \cref{lem: SZ}, we have that the number of zeroes of $\tilde{A}$ on $\F \times \F$ can be at most $Dq$, which by the choice of $D$ is at most $300/\gamma \cdot dq \leq O\left(\frac{dq}{\epsilon^2}\right)$. We also have that $|S| > \Omega(\epsilon^2 q^2)$. Thus, for any sufficiently large constants $c_1, c_2$, we have that if $q > c_1\cdot d/\epsilon^{c_2}$, then, $|S|$ exceeds  $O(dq/\epsilon^2)$, thereby implying that $\tilde{A}$ must be identically zero, which is a contradiction. Thus, $\tilde{A}$ must depend on $z$.
\end{proof}

\subsection{Proofs of \cref*{clm:NdD-lower-bound}, \cref*{clm:NdDp-lower-bound} and \cref*{clm:NdDp-lowerbound-final}}\label{sec:proofs-of-estimates-on-Ndd-etc}

\begin{proof}[Proof of \cref{clm:NdD-lower-bound}]
From the definition, we have that 
\begin{align*}
|N_{d, D}| &\geq \sum_{k = 0}^{\lfloor D/d \rfloor} \sum_{i = 0}^{(D-kd)} \sum_{j = 0}^{(D-kd-i)} 1 \, ,\\
 &\geq \sum_{k = 0}^{\lfloor D/d \rfloor} (D-kd)(D-kd + 1)/2\, , \\
 &\geq 1/2 \cdot \sum_{k = 0}^{\lfloor D/d \rfloor} (D-kd)^2 \, ,\\
 &\geq 1/2 \cdot \sum_{k = 0}^{\lfloor D/d \rfloor} (D^2 + k^2d^2 - 2Ddkd) , ,\\
 &\geq 1/2 \cdot \left( (\sum_{k = 0}^{\lfloor D/d \rfloor} D^2) + (\sum_{k = 0}^{\lfloor D/d \rfloor} k^2d^2) - (\sum_{k = 0}^{\lfloor D/d \rfloor} 2Ddk)\right)  
\end{align*}
For ease of notation, let $\tilde{D}$ denote ${\lfloor D/d \rfloor}$. So, we have $(D/d - 1) \leq \tilde{D} \leq (D/d)$. Now, from standard estimates on sums and sums or squares of the first $n$ natural numbers, we have the following. 
\[
\sum_{k = 0}^{\tilde{D}} D^2 \geq D^2\cdot (\tilde{D} + 1) \in [D^3/d, D^3/d + D^2] \, .
\]
Similarly,
\[
\sum_{k = 0}^{\tilde{D}} k^2d^2 = d^2 \cdot \left(\sum_{k = 0}^{\tilde{D}} k^2\right) \geq d^2 \cdot (\tilde{D})(\tilde{D} + 1)(2\tilde{D} + 1)/6 \, . 
\]
Plugging in the bounds for $\tilde{D}$, we get that  \[
\sum_{k = 0}^{\tilde{D}} k^2d^2  \in [1/6(2D^3/d - 3D^2 + Dd), 1/6(2D^3/d + 3D^2 + Dd)] \, .
\]
Finally, 
\[
\sum_{k = 0}^{\tilde{D}} 2Ddk = 2Dd \cdot \tilde{D}(\tilde{D}+1)/2 \, , 
\]
which, again using the upper and lower bounds on $\tilde{D}$ satisfies 
\[
\sum_{k = 0}^{\tilde{D}} 2Ddk \in [D^3/d - D^2, D^3/d + D^2] \, .
\]
Putting the estimates together, we get the claim. 
\end{proof}

\begin{proof}[Proof of \cref{clm:NdDp-lower-bound}]
For any $d, D \in \N$, and $k \in \Z_{\geq 0}$, let $A_k$ be the set defined as 
\[
A_{k} = \setdef{(i,j) \in \Z_{\geq 0}^2}{i + j + dk \leq D} \, .
\]
Clearly, the size of $A_k$ is a non-increasing function of $k$. Thus, for every integer $\ell \geq 1$ and for every $j \in \{1, 2, \ldots, p-1\}$ we have that $|A_{\ell p }| < |A_{(\ell-1)p + j}|$, thereby implying that 
\[
|A_{\ell p }| < \frac{1}{p}\left(\sum_{j = 1}^{p} |A_{(\ell-1)p + j}| \right) \, .
\] 
Summing over $\ell$, we get that 
\[
\sum_{\ell = 1}^{\lfloor D/d \rfloor}|A_{\ell p }| < \frac{1}{p} \left(\sum_{\ell = 1}^{\lfloor D/d \rfloor} \sum_{j = 1}^{p} |A_{(\ell-1)p + j}| \right)\, .
\] 
We now note from the definitions of the sets $N_{d, D}$ and $A_k$ that $|N_{d, D}| = \sum_{k = 0}^{\lfloor D/d \rfloor} |A_k|$ and, thus, we have that
\[
\sum_{\ell = 1}^{\lfloor D/d \rfloor}|A_{\ell p }| < \frac{1}{p} |N_{d, D}| \, .
\]
Moreover, the size of the set $N_{d, D, p}$ satisfies 
\[
|N_{d, D, p}| \geq |N_{d, D}| - \sum_{\ell = 1}^{\lfloor D/d \rfloor}|A_{\ell p }|  \geq \left(1-\frac{1}{p}\right)|N_{d, D}| \, , 
\]
which, for $p \geq 2$ gives $|N_{d, D, p}| \geq \frac{1}{2}|N_{d, D}|$. 
\end{proof}

\begin{proof}[Proof of \cref{clm:NdDp-lowerbound-final}]
From the lower bound on $|N_{d, D}|$ in \cref{clm:NdD-lower-bound}, we have that $|N_{d, D}| \geq D^3/3d - 5/2D^2$. Now, if $D/d > 20$, then  $|N_{d, D}| \geq D^2(D/3d - 5/2) \geq D^3/6d$. Combining this with the bound in \cref{clm:NdDp-lower-bound} completes the proof. 
\end{proof}

\section{Variations on Low-degree Testing}\label{sec:LDT-variants}

In this section we show the equivalence of some standard variations of low-degree testing. These results are essentially folklore - our proofs follow closely the proofs from Arora and Sudan~\cite[Section 2.2]{AroraS2003}, confirming along the way that the proofs continue to hold even when the field size is only linear in the degree.

For a function $f \colon \F^m \to \F$ and line $\ell$ in $\F^m$, let $P^{(f,d)}_\ell$ be the best fit degree $d$ univariate polynomial for $f$ on the line $\ell$. We now, define the following quantities for any function $f\colon \F^m \to \F$, line $\ell$ and plane $\pi$
\begin{align*}
    \delta_f(\ell) &:= \Pr_{x \in \ell}[P^{(f,d)}_\ell(x) \neq f(x)], \\
    \delta_f(\pi)  & := \E_{\ell \in \pi}[\delta_f(\ell)] = \Pr_{\substack{\ell \in \pi\\x \in \ell}}[P^{(f,d)}_\ell(x) \neq f(x)],\\
    \delta_f &:= \E_\pi[\delta_f(\pi)] = \Pr_{\substack{\ell \in \F^m\\x \in \ell}}[P^{(f,d)}_\ell(x) \neq f(x)].
\end{align*}

As mentioned in \cref{rem:lines-table-version}, the general low-degree test is provided two oracles $f:\F_q^m \to \F_q$ and $P: \calL^{(m)} \to \F_q^{\leq d}[t]$, the first mapping points in $\F_q^m$ to values in $\F_q$ and the latter mapping lines in $\F_q^m$ to (a table of evaluations of) a degree $\leq d$ univariate polynomial. The probability that the low-degree test accepts is given by $\Pr_{x, \ell \ni x}[f(x) = P[\ell](x)]$.  

We say $x$ is $\beta$-good for $(f,P)$ if $\Pr_{\ell \ni x}[P[\ell](x) = f(x)] \geq \beta$. Note that if the low-degree test accepts $(f,P)$ with probability $\beta$ then at least $(\beta/2)$ fraction of the points are $(\beta/2)$-good for $(f,P)$. Conversely, if the low-degree test accepts $(f,P)$ with probability at most $\beta = \beta_1\beta_2$ then at most $\beta_1$ fraction of the points are $\beta_2$-good for $(f,P)$ (for any choice of $\beta_1,\beta_2$ satisfying $\beta = \beta_1\beta_2$). 

\begin{definition}[Weak form of LDT] \label{def:weak-ldt}
    Given field $\F_q$, integer parameters $m \geq 1$, $d \geq 0$, real $\beta > 0$ and function $h:\R^{>0} \to \R^{>0}$, we say that {\em weak low-degree testing} holds for
    $(\F_q,m,d,\beta_0,h)$ if the following is true:
    \begin{quote}
        For all $f: \F_q^m \to \F_q$ and $P: \calL^{(m)} \to \F_q[t]^{\leq d}$ and all $\beta > \beta_0$, if $(f,P)$ pass the low-degree test with probability at least $\beta$, then there exists some degree $d$ polynomial $Q$ such that $\agree(f,Q) \geq h(\beta)$.
    \end{quote}
\end{definition}

Note that \cref{thm:LDT-low-agreement-poly-agreement} asserts that there are universal constants $C, \alpha$ such that weak low-degree testing holds for $(\F_q,2,d,\beta_0,h)$ for $h:\beta \mapsto \alpha \beta^4$ provided $q \geq Cd/\beta_0^7$. 

\medskip

There are two incomparable ways to strengthen a weak low-degree test and we first define the \emph{list-decoding} variant.

\begin{definition}[List-decoding form of the LDT]
    \label{def:list-decoding-ldt}
    Given a field $\F_q$, integer parameters $m \geq 1$, $d \geq 0$, real $\beta_0 > 0$ and function $h_1:\R^{>0} \to \R^{>0}$, we say that {\em list-decoding low-degree testing} holds for
    $(\F_q,m,d,\beta_0,h_1)$ if the following is true:
    \begin{quote}
        For all $\beta > \beta_0$, $f:\F_q^m \to \F_q$ and $P: \mathcal{L}^{(m)}\rightarrow \F_d[t]^{\leq d}$, there is a (possibly empty) list of at most $C \leq \frac{2}{h_1(\beta)}$ $m$-variate degree $d$ polynomials $Q_1,\ldots, Q_C$ such that $\agree(f,Q_i) \geq h_1(\beta)$ for each $i$, and 
        \[
            \Pr_{x\in \F_q^m}\insquare{\text{$x$ is $\beta$-good for $(f,P)$ and $f(x)\notin \set{Q_1(x), \ldots, Q_C(x)}$}} \leq \beta_0.
        \]
    \end{quote}
    The above must hold regardless of the LDT acceptance probability of $(f,P)$. Indeed, if $(f,P)$ passes the LDT with probability less than $\beta \cdot \beta_0$, then the list of polynomial may be empty as there can be at most $\beta_0$ fraction of $\beta$-good points for $(f,P)$. 
\end{definition}

\noindent Another strengthening of the weak LDT is in terms of the agreement probability. 

\begin{definition}[High-agreement form of the low-degree test]
    \label{def:high-agreement-ldt}
Given a finite field $\F_q$, integer parameter $m \geq 1$, $d \geq 0$, real $\beta_0 > 0$ and a function $h_2:\R^{>0} \to \R^{>0}$, we say that $(\F_q, m, d, \beta_0, h_2)$ if the following is true:
\begin{quote}
    For every $\beta > \beta_0$, $f:\F_q^m \rightarrow \F_q$ and $P:\mathcal{L}^{(m)} \rightarrow \F_q[t]^{\leq d}$ such that $(f,L)$ passes the low-degree test with probability at least $\beta$, then there exists an $m$-variate degree $d$ polynomial $Q$ such that $\agree(f,Q) \geq \beta - h_2(\beta_0)$.
\end{quote}
\end{definition}

In our applications we will assume $h,h_1,h_2$ are monotone non-decreasing functions. We will see that a weak low-degree test with $h(\beta) \to 0$ as $\beta \to 0$ imply both the list-decoding variant and the high-agreement variant for some appropriate functions $h_1, h_2$ that also satisfy $h_1(\beta), h_2(\beta) \to 0$ as $\beta\to 0$. 

In our applications we will assume $h,h_1,h_2$ are monotone non-decreasing functions. Further we assume $h(\beta) \to 0$ as $\beta \to 0$ and show that $h_1(\beta) \to 0$ and $h_2(\beta) \to 0$ as $\beta \to 0$. 
Note that the two implications above are incomparable and do not directly imply each other. However as we will see in the proof, the implication in (1) is useful to prove (2). 

\begin{lemma}[Weak LDT implies list-decoding LDT]
    \label{lem:weak-LDT-implies-list-decoding-LDT}
    Suppose $\F_q, m, d, \beta_0$ and $h(\cdot)$ are such that  weak low-degree testing holds for $(\F_q, m, d, \beta_0, h)$. Then list-decoding low-degree testing holds for $(\F_q, m, d, \beta_0', h_1)$ for 
    \begin{itemize}\itemsep0pt
        \item $\beta_0'$ satisfying $\beta_0' > \sqrt{\beta_0}$ and $h(\beta_0'^2) \geq \max\inparen{\frac{e}{\sqrt{4q}, \frac{2(d+1)}{q},\sqrt{\frac{d}{q}}}}$,
        \item $h_1:\beta \mapsto \sfrac{1}{2} \cdot h(\beta \cdot \beta_0)$.
    \end{itemize}
\end{lemma}
\begin{proof}
    Fix any $\beta > \beta_0'$. Let $Q_1,\ldots, Q_C$ be the set of all $m$-variate degree $d$ polynomials that have agreement at least $\eta = \sfrac{1}{2}\cdot h(\beta \cdot \beta_0') > \sfrac{1}{2}\cdot h(\beta_0'^2)$ with $f$ on $\F_q^m$. By the Johnson bound, we have that $C \leq \sfrac{2}{\eta}$ since $\eta \geq 2\sqrt{\sfrac{d}{q}}$ by the choice of $\beta_0'$. 

    Consider the following randomly chosen function $g:\F_q^m \to \F_q$ given by $g(x) = f(x)$ if $x \not\in \cup_i S_i$ and $g(x) \sim \text{Unif}(\F_q)$ if $x \in \cup_i S_i$. \cref{clm:random-g} below asserts that with positive probability we have that $\agree(g,Q) < 2\eta$ for every degree $d$ polynomial $Q$. Fix a $g$ such that this holds. By the weak low-degree test applied to $g$ (in contrapositive form) it follows that $(g,P)$ pass the low-degree test with probability at most $\beta \cdot \beta_0'$ since $\beta \cdot \beta_0' > \beta_0'^2 \geq \beta_0$ and $h(\beta \cdot \beta_0') = 2\eta$. Therefore, $\Pr_{x\in \F_q^m}[\text{$x$ is $\beta$-good for $(g,L)$}] < \beta_0'$. But now note that $x \notin \cup_i S_i$ satisfies $f(x) = g(x)$ and so such an $x$ is $\tau$-good for $(g,P)$ iff it is $\beta$-good for $(f,P)$. We conclude that 
    \begin{align*}
       \hspace{2em}&\hspace{-2em}\Pr_{x\in \F_q^m}[\text{$x$ is $\beta$-good for $(f,L)$ and $f(x) \notin \{P_1(x),\ldots,P_C(x)\}$ }] \\
       & = \Pr_{x\in \F_q^m}[\text{$x$ is $\beta$-good for $(g,L)$ and $f(x) \notin \{P_1(x),\ldots,P_C(x)\}$ }] \\
       & = \Pr_{x\in \F_q^m}[\text{$x$ is $\beta$-good for $(g,L)$}] \leq \beta_0'.
    \end{align*}

    \noindent
    To complete the proof, we only need to prove the following claim. 

    \vspace{-2em}

    \begin{quote}
        \begin{claim}\label{clm:random-g}
            If $\eta \geq \max\{\frac{e}{\sqrt{q}},\frac{4(d+1)}{q}\}$ then $$\Pr[ \exists Q:\F_q^m\to\F_q\text{ , }\deg(Q) \leq d\text{ , }\agree(g,Q) \geq 2\eta ] \leq q^{-(\eta/4)q^m} < 1.$$
        \end{claim}

        \begin{proof}
            Fix a degree $d$ polynomial $Q$. If $Q \notin \{Q_1,\ldots,Q_C\}$ then $\agree(Q,f) < \eta$ and if $Q \in \{Q_1,\ldots,Q_C\}$ we have for every 
            $x \notin \cup_i S_i$, $Q(x) \ne f(x) = g(x)$. Thus in either case we have $|\{x \notin \cup_i S_i | Q(x) = g(x) = f(x) \}| \leq \eta q^m$. Thus to have agreement at least $2\eta$ with $Q$, $g$ must satisfy $|\{x \in \cup_i S_i | P(x) = g(x) \}| \geq \eta q^m$. We show below that the probability that this happens is at most $q^{-(\eta/2)q^m}$. 

            For a fixed set $S \subseteq \cup_i S_i$ of size $\eta q^m$, the probability that $g$ and $Q$ agree on the set is $q^{-\eta q^m}$. The number of sets $S$ is at most $\binom{q^m}{\eta q^m} \leq (e/\eta)^{\eta q^m} \leq q^{(\eta/2) q^m}$ (using $\eta \geq e/\sqrt{q}$). We conclude that the probability that there exists a set $S \subseteq \cup_i S_i$ of size at least $\eta q^m$ such that $g$ and $Q$ agree on $S$ is at most $q^{-(\eta/2)q^m}$. 

            Now to conclude the proof we take a union bound over all $Q$'s. The number $m$ variate monomials of degree at most $d$ is clearly at most $(d+1)^m$ and so the number of polynomials is at most $q^{(d+1)^m} \leq q^{((\eta/4) q)^m} \leq q^{(\eta/4) q^m}$ (where the first inequality uses $\eta \geq 4(d+1)/q$ and the second uses $m \geq 1$). We conclude that the probability that there exists $Q$ of degree at most $d$ such that 
            $|\{x \in \cup_i S_i | Q(x) = g(x) \}| > \eta q^m$ is at most $q^{-(\eta/4)q^m}$. The claim follows.
        \end{proof}
    \end{quote}
    \noindent
    This completes the proof of \cref{lem:weak-LDT-implies-list-decoding-LDT}. 
\end{proof}

\begin{lemma}[Weak LDT implies high-agreement LDT]
    \label{lem:weak-LDT-implies-high-agreement-LDT}
    Suppose $\F_q, m, d, \beta_0$ and $h(\cdot)$ are such that  weak low-degree testing holds for $(\F_q, m, d, \beta_0, h)$. Then, high-agreement low-degree testing holds for $(\F_q, m, d, \beta_0'', h_2)$ for 
    \begin{itemize}\itemsep0pt
        \item $\beta_0''$ satisfying $\beta_0'' > \beta_0'$ and $\beta_0''^3 \cdot h_1(\beta_0''^2) > \sfrac{2}{q}$ and $\beta_0'' \cdot h_1(\beta_0''^2) \geq \sfrac{2d}{q}$,
        \item $h_2:\beta \mapsto 3\beta$,
    \end{itemize}
    where $\beta_0'$ and $h_1(\cdot)$ are as implied by \cref{lem:weak-LDT-implies-list-decoding-LDT}.

    In other words, if $(f,P)$ passes the LDT with probability $\beta > \beta_0''$, then there is some $m$-variate degree $d$ polynomial $Q$ such that $\agree(f,Q) \geq \beta - 3\beta_0''$. 
\end{lemma}
\begin{proof}
    By applying \cref{lem:weak-LDT-implies-list-decoding-LDT} (with $\beta = \beta_0'' \geq \beta_0'$) we have that for any pair $(f,P)$, there exist at most $C = 2/h_1(\beta_0'')$ polynomials $Q_1,\ldots, Q_C$, each with agreement at least $h_1(\beta_0'')$ with $f$ such that
    \[
        \Pr_{x\in \F_q^m}\insquare{\text{$x$ is $\beta_0''$-good for $(f,P)$ and $f(x) \notin \set{Q_1(x),\ldots, Q_C(x)}$}} \leq \beta_0'.
    \]
    Let $S_i = \setdef{x}{Q_i(x)=f(x)}$ and let $\eta_i = q^{-m}|S_i|$. Assume w.l.o.g. that $\eta_1 \geq \cdots \geq \eta_C$. Note that we are not guaranteed that $C \geq 1$ and the above list of polynomials be empty. However, we wish to show that if $(f,P)$ passes the LDT with probability $\beta \geq \beta_0''$, then $C \geq 1$ and $\eta_1 \geq \beta - \gamma$, where $\gamma = 3\beta_0''$. We will do so by proving that the acceptance probability of the low-degree test on $(f,P)$ is upper bounded by $\eta_1 + \gamma$. 

    Define a line $\ell$ to be {\em standard} if $P[\ell] \in \{Q_1|_\ell,\ldots,Q_C|_\ell\}$. Say the $\ell$ is {\em abnormal} if there exists $i \in [C]$ such that $\abs{\setdef{x \in \ell}{Q_i(x) = f(x)}} \geq (\eta_1 + \beta_0'')\cdot q$. For a non-standard line $\ell$ say that a point $x\in \ell$ is {\em coincidental} for $\ell$ if there exists $i \in [C]$ s.t. $Q_i(x) = P[\ell](x)$. Finally say that a pair $(x,\ell)$ with $x \in \ell$ is {\em unexplained} if $f(x) = P[\ell](x)$ and $f(x) \not\in \{Q_1(x),\ldots,Q_C(x)\}$. 

    We now upper bound the probability that the low-degree test accepts a random pair $(x,\ell)$ by consider various cases. We first note that for the low-degree test to accept a pair $(x,\ell)$ at least one of the following must happen: 
    \begin{enumerate}\itemsep0pt
    \item\label{item:weak-to-strong:standard-normal} $\ell$ is standard and normal and the low-degree test accepts, or
    \item\label{item:weak-to-strong:abnormal} $\ell$ is abnormal, or
    \item\label{item:weak-to-strong:coindidental} $\ell$is non-standard and $x$ is coincidental for $\ell$, or
    \item\label{item:weak-to-strong:unexplained} $(x,\ell)$ is unexplained.
    \end{enumerate}
    The typical case is \cref{item:weak-to-strong:standard-normal} where $\ell$ is standard and normal. In this case the probability over $x$ that the low-degree test accepts the pair $(x,\ell)$ is at most $\eta_1 + \beta_0''$. 

    For \cref{item:weak-to-strong:abnormal}, the probability that a random line $\ell$ is abnormal is upper bounded by $\gamma_1 := C/(\gamma^2 q)$ by a Chebychev argument (for fixed $i\in C$ the expected fraction of agreement is $\eta_i$ and a random line contains $q$ pairwise independent random samples of points from $\F_q^m$). Since $\beta_0''$ satisfies $\beta_0''^3 \cdot h_1(\beta_0''^2) > \sfrac{2}{q}$, we this probability is bounded by $\beta_0''$. 

    For \cref{item:weak-to-strong:coindidental}, the probability that a point $x$ on a non-standard line $\ell$ is coincidental for the line is at most $Cd/q$ (for every $i\in [C]$ there are at most $d$ points where $Q_i(x) = P[\ell](x)$). Since $\beta_0''$ satisfies $\beta_0'' \cdot h_1(\beta_0''^2) > \sfrac{2d}{q}$, we this probability is bounded by $\beta_0''$ as well.

    And finally for \cref{item:weak-to-strong:unexplained}, the probability that a pair $(x,\ell)$ is unexplained is, by the list-decoding version of the LDT, at most $\beta_0' \leq \beta_0''$. 

    We thus conclude that $(f,P)$ passes the low-degree test accepts with probability at most $\eta_1+3\beta_0''$.
\end{proof}

\end{document}